\documentclass[12pt,a4paper,authoryear]{elsarticle}
\usepackage[centertags]{amsmath}
\usepackage[english]{babel}
\usepackage{times}
\usepackage{amsfonts}
\usepackage{amssymb}
\usepackage{amsthm}
\usepackage{graphicx}
\usepackage{caption}
\usepackage{subcaption}
\usepackage{multicol}
\usepackage{multirow}
\usepackage{lscape}
\usepackage{esint}
\usepackage{bm}
\usepackage{color}
\usepackage[cp1250]{inputenc}
\usepackage{epstopdf}
\usepackage{miller}
\usepackage{float}

\newtheorem{theorem}{Theorem}[section]

\newtheorem{proposition}{Proposition}[section]
\newtheorem{remark}{Remark}[section]

\DeclareMathOperator{\sym}{sym}

\title{Branching of twins in shape memory alloys revisited}
\author[1]{Hanu\v{s} Seiner}
\author[2]{Paul Plucinsky}
\author[2,3]{Vivekanand Dabade}
\author[1,4]{Barbora Bene\v{s}ov\'{a}}
\author[2]{Richard D. James}
\address[1]{Institute of Thermomechanics, Czech Academy of Sciences, Prague, Czech Republic}
\address[2]{Aerospace Engineering and Mechanics, University of Minnesota, Minneapolis, USA}
\address[3]{Ecole Polytechnique, University of Paris-Saclay, France}
\address[4]{Faculty of Mathematics and Physics, Charles University, Prague, Czech Republic.}

\begin{document}


\begin{abstract}

\baselineskip=20pt

We study the branching of twins appearing in shape memory alloys at the interface between austenite and martensite.  In the framework of three-dimensional non-linear elasticity theory, we propose an explicit, low-energy construction of the branched microstructure, generally applicable to any shape memory material without restrictions on the symmetry class of martensite or on the geometric parameters of the interface. We show that the suggested construction follows the expected energy scaling law, i.e., that (for the surface energy of the twins being sufficiently small) the branching leads to energy reduction. 
Furthermore, the construction can be modified to capture different features of experimentally observed microstructures without violating this scaling law. By using a numerical procedure, we demonstrate that the proposed construction is able to predict realistically the twin width and the number of branching generations in a Cu-Al-Ni single crystal.
  
\end{abstract}

\journal{J. Mech. Phys. Solids}
\maketitle

\baselineskip=20pt

\section{Introduction}

 In  shape memory alloys, the branching of the ferroelastic domains, called martensitic twins, typically appears close to the interface between austenite and a first order laminate of two martensitic variants (Fig.~\ref{branchingEx}). While the gradual refinement of the twins towards the interface reduces the elastic strain energy localized directly at the interface, coarsening of the laminate farther from the interface leads to reduction of the surface energy in the crystal. As a result, the elastic energy is partially delocalized from the interface into the branched structure, while the surface energy becomes partially localized into the vicinity of the interface, which both may lead to reduction of the total energy. 

 \begin{figure}[!h]
 \centering
 \includegraphics[width=0.75\textwidth]{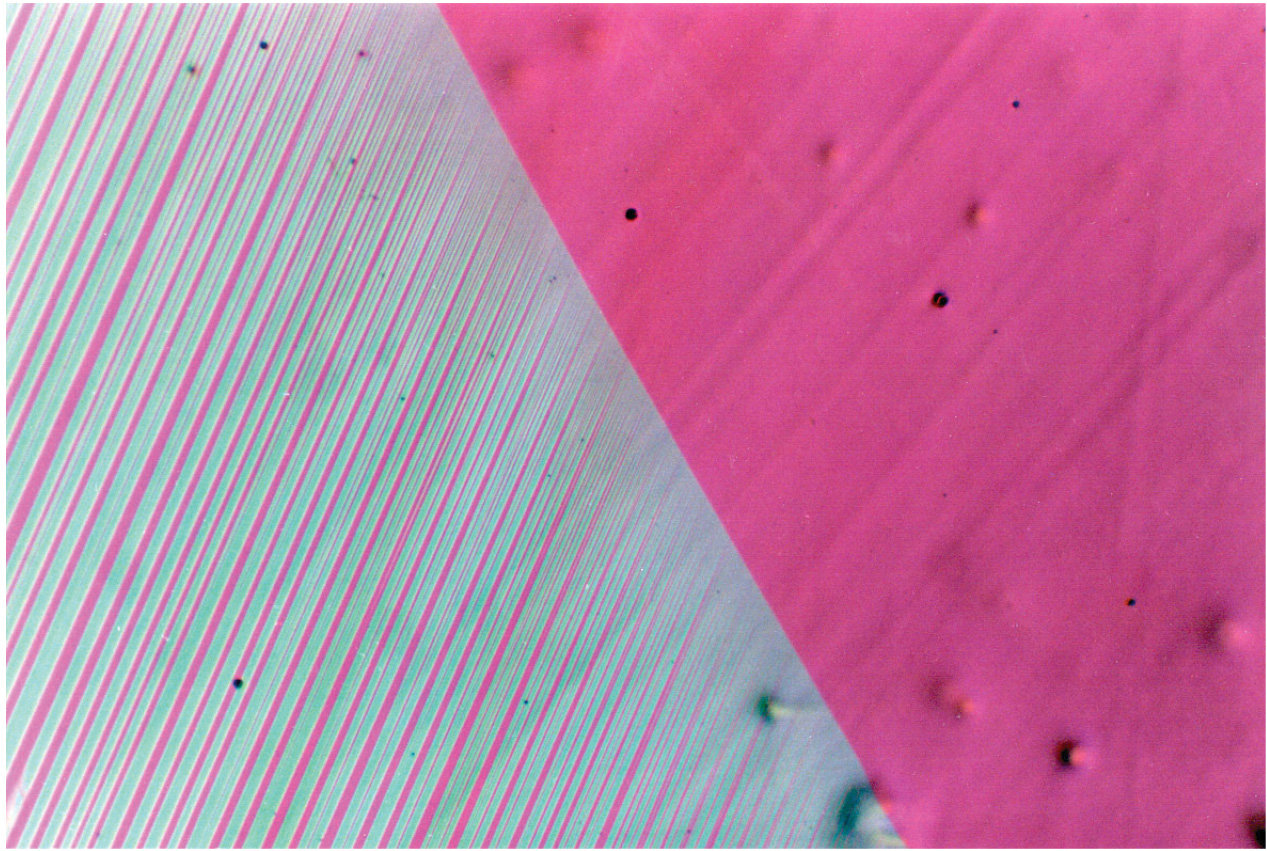}
 \caption{An example of branching microstructure at the austenite-martensite interface.  The photograph is about $0.5$mm in length.  Courtesy of C. Chu. } \label{branchingEx}
\end{figure}

The theoretical framework for studying the branched microstructures in shape memory alloys was first established by \cite{KM1,KM2}, and further developed by many others \citep{Capella_Otto_1,Capella_Otto_2,Dondl,Conti}. These pioneering works always considered some simplified versions of the problem, mostly assuming linearized elasticity and/or prescribing some artificial continuous displacement fields over the branched structure. Most recently, the full vectorial problem, including invariance under rotations, was studied in two dimensions by  \cite{ChanConti1,ChanConti2}, assuming specific forms of the deformation gradients representing individual variants of martensite. Although all these works provide a valuable insight into the scaling laws for energy of the branched structure, they are hardly accessible for experimentalists, who would like get a direct understanding of the branching mechanism for a given, real material.  

In this paper, we present a construction of branched microstructure at the austenite-martensite interface in a fully non-linear, three-dimensional setting, and we show that  branching using this construction leads to the expected reduction of the total energy. In addition, this explicit construction enables us to discuss the character of the elastic strain field in the branched microstructure, and to suggest strategies for reducing the energy stored in this microstructure. The main aim of this paper is to provide a micro-mechanical model of branching in shape memory alloys, directly applicable to particular shape memory alloys with experimentally determined material parameters. Let us point out that the energy calculated for the construction presented in this paper is still just an upper bound to the energy of the real branched microstructure. Similarly to \citep{KM1,KM2,Capella_Otto_1,Capella_Otto_2,ChanConti1,ChanConti2}, we do not require stress equilibrium at the interfaces inside the construction, and the number of degrees of freedom of the construction is relatively low. However, as all parts of the constructed microstructure represent either stress-free martensite or martensite very close to the stress-free state, and as the geometric parameters of the construction are optimized in order to attain a minimum of energy, it is justified to consider this construction as a good approximation of a real microstructure. This assumption is tested in the final part of the paper on a numerical example, where we show that the simulation predicts  realistically the length-scales and morphology of an austenite-to-martesite interface in the Cu-Al-Ni alloy.

\section{Theoretical Background}\label{sec:Theory}

\subsection{A simple model of the austenite-martensite interface}

A widely accepted theoretical approach to martensitic microstructures \citep{James_Ball_1,James_Ball_2,Bhattacharya} takes the advantage of describing the diffusionless transitions in shape memory alloys within the framework of continuum mechanics. In this theory, the austenite phase and the individual variants of martensite are represented by 3$\times{}$3 matrices, with the identity matrix ${\mathbf I}$ representing the austenite phase, and Bain matrices ${\mathbf U}_i$ for the variants of martensite. ${\mathbf U}_i$ are calculated from the deformation gradients $\nabla{\mathbf y}(\mathbf x)$ that are related to the crystal lattices of the variants via the Cauchy-Born hypothesis \citep{Bhattacharya}. For $\varphi(\nabla{\mathbf y})$ being the free energy density of the shape memory alloy crystal (occupying some purely austenite domain $\Omega$ in the reference configuration), the theory predicts that the martensitic microstructure forming at a given temperature and for given boundary conditions at $\partial{}\Omega$ corresponds to the minimum of the energy
\begin{equation}
 E=\int_{\Omega}\varphi(\nabla{\mathbf y}) {\rm d}{\mathbf x} + E_{\rm surf.}[{\mathbf y}]\label{energy}
\end{equation}
over all continuous functions ${\mathbf y}({\mathbf x})$, where the second term represents the energetic penalization for the interfaces between the individual phases in the observed domain. The energy density $\varphi(\nabla{\mathbf y})$ is typically considered to have a multi-well structure, with the respective multiple minima corresponding to the individual phases and variants; the depth of the energy wells gives the chemical energy of the phases at the given temperature, and the derivatives $\tfrac{\partial^2 \varphi}{\partial F_{jk} \partial F_{lm}}$ evaluated at $\mathbf{I}$ or $\mathbf{U}_{i}$ give the elastic constants of these phases. At the transitions temperature, i.e.,
at the temperature where the chemical energy of the austenite and the martensite phases are equal,
\begin{equation}
\varphi(\nabla{\mathbf y}\in{SO(3)})= \varphi(\nabla{\mathbf y}\in\cup_{i=1}^n{SO(3)}{\mathbf U}_{i}) = 0,
\end{equation}
where ${\mathbf U}_{i}$ is the Bain matrix representing the $i-$th variant of martensite from the total number of $n$ possible martensitic variants.

The surface energy term $E_{\rm surf.}[{\mathbf y}]$ is more delicate. Several different expressions for this term can be found in the literature. A frequently used form is
\begin{equation}
 E_{\rm surf.}[{\mathbf y}]=\int_{\Omega}\sigma{}|\nabla^2{\mathbf y}| {\rm d}{\mathbf x}, \label{surfE}
\end{equation}
where $\sigma$ has a meaning of the surface energy per unit area per unit jump of the strain over the interface; the minimum of (\ref{energy}) is then sought over all continuous ${\mathbf y}({\mathbf x})$ with measurable (weak) second derivative.  
However, relating the value of the surface energy per unit area to the magnitude of the strain jump over the interface does not have any clear physical justification. For example, in  alloys undergoing  cubic-to-orthorhombic or cubic-to-monoclinic transitions, certain pairs of variants can form at the same time two different types of twins: Type I twins, where the twinning plane is a low-index atomic plane, and the Type II twins, where the twinning plane is a general plane with irrational crystallographic indices. As a result, the Type I interface can be atomistically sharp, while Type II interface can have a diffused, or segmented nature \citep{Vronka,Liu1,Liu2}. Hence, although the strain jump over the twinning plane is nearly identical for both twinning types, the surface energy associated with the twin can be very different due to the different structure of the interface. 

Another possible approach is to consider a fixed value $\sigma_0$ of surface energy per unit area  for a given twin, regardless of the jump in strain between the variants forming the twin. In this approach, though, it is not clear how $\sigma_0$ evolves when the twinning plane becomes inclined, as in Fig.~\ref{branched}. Again, inclining a Type I interface, which is fixed to a crystallographic plane, could lead to a very different change of the surface energy than in the case of Type II, which can slightly rotate without changing its structure. 

In this paper, we will use both these approaches, depending on the particular case. For the qualitative models of the simple laminate and of the branched structure (this section and Section 3), where the inclination of the twins can be neglected, we will assume a constant surface energy per unit area. The fully-optimized construction of the branched structure used in Section 4 for quantitative simulations of a real shape memory crystal requires, however, a treatment of inclined interfaces. In this case, we will use the assumption (\ref{surfE}).

\bigskip
 Let ${\mathbf U}_{A}$ and ${\mathbf U}_{B}$ be two variants of martensite able to form a twin, i.e., able to border over a planar, kinematically compatible interface.  This means that there exists a rotation matrix ${\mathbf R}$ and a non-zero vector $\hat{\mathbf a}$ such that
\begin{equation}
{\mathbf R}{\mathbf U}_{A}-{\mathbf U}_{B}=\hat{\mathbf a}\otimes{\mathbf n},
\end{equation}
where ${\mathbf n}$ is a unit vector perpendicular to the twinning plane. Let us further consider that a first order laminate of these two variants is able to form a macroscopically compatible interface with austenite, i.e., that there exists a rotation matrix ${\mathbf Q}$, a non-zero vector ${\mathbf b}$ and a volume fraction $0<\lambda<{1}$ such that 
\begin{equation}
{\mathbf Q}\left[\lambda{\mathbf R}{\mathbf U}_{A}+(1-\lambda){\mathbf U}_{B}\right]-{\mathbf I}={\mathbf b}\otimes{\mathbf m}, 
\end{equation}
where ${\mathbf m}$ is a unit vector perpendicular to the habit plane. (We will not discuss here the limiting cases $\lambda\rightarrow{0}$ and $\lambda\rightarrow{1}$, where the compatibility conditions are approximately satisfied for a single variant of martensite \citep{ContiZwicknagl,Zwicknagl}.)
By introducing a simplified notation $\mathbf{A}={\mathbf Q}{\mathbf R}{\mathbf U}_A$, $\mathbf{B}={\mathbf Q}{\mathbf U}_B$ and ${\mathbf a}={\mathbf Q}\hat{\mathbf a}$, we obtain the basic set of compatibility equations used throughout this paper:
\begin{align}
&\mathbf{A}-\mathbf{B}={\mathbf a}\otimes{\mathbf n}, \label{compat1} \\
&\lambda{}\mathbf{A}+(1-\lambda)\mathbf{B}-\mathbf{I} ={\mathbf b}\otimes{\mathbf m}. \label{compat2}
\end{align}  

\bigskip
Consider now a simple $AB$ laminate (without branching) forming a planar interface with austenite.  The macroscopic compatibility condition (\ref{compat2}) ensures that the deformation gradient of the twinned martensite, in the limit of an infinitely fine $AB$ laminate, is compatible with the austenite phase.  This property enables a nearly stress-free co-existence between the austenite and martensite phases at a habit plane with normal $\mathbf{m}$. However, the laminate cannot be infinitely fine, as fine oscillations in $\nabla{\mathbf y}$ would result in diverging surface energy. In contrast, any finite width $d$ of  a single $AB$ twin introduces elastic strains to ensure compatibility at the austenite-martensite interface. Hence, the energy of such an interface is non-zero and must be obtained from (\ref{energy}) by balancing the elastic and surface energy. 

Following the approach of \cite{James_Ball_1}, the minimization (\ref{energy}) for a simple laminate meeting with austenite can be done quite easily: Whatever the elastic strain field is at the habit plane, the corresponding strain energy for one twin approaching this interface scales as $d^2$ (due to the scale-independent character of the linear elasticity). The interfacial energy is proportional to the number of twins.  For $1/d$ being the number of twins per unit length of the habit plane, and for $L$ being the distance between the habit plane and the free surface of the crystal (assumed parallel to the habit plane), the total energy per unit length of the habit plane is
\begin{equation}
 E=\frac{1}{d}\left({G_{AB}}{d^2}+L\sigma_{AB}\right)=d{G_{AB}}+\frac{L\sigma_{AB}}{d}, \label{laminate}
\end{equation}
where $\sigma_{AB}$ is a specific surface energy paid for a unit area of the twin interface, and $G_{AB}$ is a constant characterizing the elastic energy paid for one twin with width $d=1$ approaching the habit plane. This constant is, in principle, a combination of the elastic constants of austenite and martensite, the transformation strains, the volume fraction $\lambda$ and other geometrical parameters of the habit plane. The optimal $d$ for a simple laminate is then  
\begin{equation}
 d=\sqrt{\frac{L\sigma_{AB}}{G_{AB}}},
\end{equation}
i.e., $d\sim{}\sqrt{L}$ which is in a good agreement with several experimental observations. The scaling of the energy (\ref{energy}) is then $E\sim{}\sqrt{\sigma_{AB}L}$ (which is vanishing in the limit $\sigma_{AB} \rightarrow 0$). The constant $G_{AB}$ can be obtained by numerical simulations, using, for example, finite element calculations \citep{stupk}, or it can be estimated (or, more precisely, bounded above) by constructing a piecewise homogeneous strain fields compensating the incompatibility at the interface \citep{James_Ball_1}. The second approach is more straightforward, and will be adopted in this paper not only for the closure domains, but for the whole branched microstructure. 

\subsection{The self-similar construction of branching of twins}

While the energy calculation for a simple $AB$ laminate is instructive, many experimental observations reveal that laminates at the austenite-martensite interfaces tend to 
branch into a finer and finer structure near the interface. The theoretical treatment of the branched structure is obviously more intricate than for a simple laminate. In the branched structure, the elastic strain energy is delocalized from the interface, as the branching requires slight inclinations of the twinning planes from the stress-free orientations.  In contrast, the surface energy becomes localized near the interface, as the number of twinning planes increases with branching. Consequently, the simple scaling model presented above does not hold and a more detailed construction is necessary.

Here, we follow the ansatz of \cite{KM1,KM2} that the branched microstructure can be constructed in a self-similar manner.  This entails constructing first a branching segment (or a cell) that provides a refinement of the number of twins per length of the segment from $1/d$ to $2/d$, and then repeating this segment at finer and finer spatial scales in geometric progression, until the required hierarchical structure is obtained. The energy (\ref{energy}) calculated for the resulting structure is then the upper bound of the  total energy of a real branched twins. 

\begin{figure}[!h]
 \centering
 \includegraphics[width=0.75\textwidth]{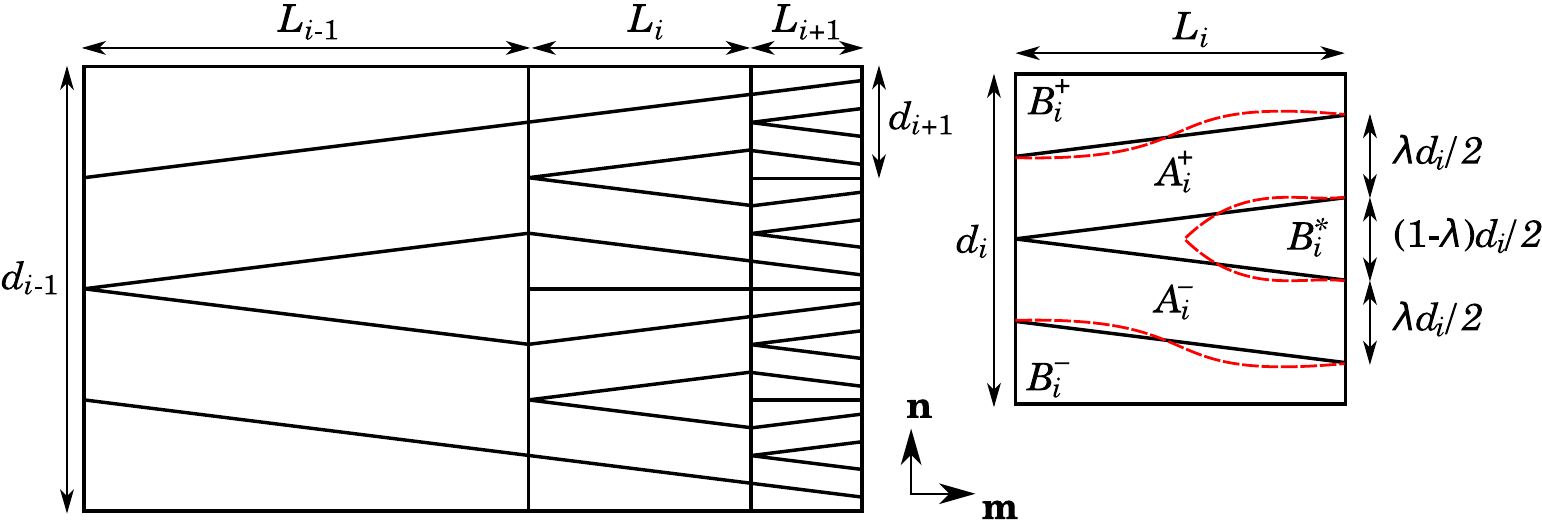}
 \caption{Self-similar construction of the branched microstructure suggested by  \cite{KM1,KM2}; for simplicity, ${\mathbf m}$ is chosen perpendicular to ${\mathbf n}$. The whole structure  is composed of individual branching segments (one segment of the $i-$th layer is shown on the right). $\mathbf{A}^{+}_i$, $\mathbf{A}^{-}_i$,$\mathbf{B}^{+}_i$, $\mathbf{B}^{-}_i$ and $\mathbf{B}^{*}_i$, are the deformation gradients in the $i-$th layer chosen such that the segments can be compatibly attached to each other and such that the compatible interfaces inside the segment are inclined as required. It is assumed that these deformation gradients differ from the deformation gradients $\mathbf{A}$ and $\mathbf{B}$, respectively, just by small perturbations, such that resulting elastic energy is small. The deformation gradients can be homogeneous, which means the interfaces between them are planar (solid black lines), or heterogeneous, which may lead to curved interfaces (dashed red lines).} \label{hierarchical_KM}
\end{figure}

The construction is outlined in Fig.~\ref{hierarchical_KM}. The segment of length $L_i$, width $d_i$ and unit thickness in the out-of-plane direction provides the branching of the laminate such that the spacing of the twins on the left-hand-side boundary of the segment is equal to $\lambda{}d_i$ and the spacing of the twins on the right-hand-side boundary is $\lambda{}d_i/2$. The interfaces connecting the the left-hand-side and right-hand-side boundaries of the segment may be either planar (solid line in Fig.~\ref{hierarchical_KM}) or curved (dashed lines in Fig.~\ref{hierarchical_KM}); in agreement with \cite{KM1,KM2}, we will consider planar interfaces for our construction. However, as discussed in Section 3, using curved interfaces does not lead to any significant reduction of the energy or any change in the scaling laws. 

If we assume that there are $1/d_0$ twins per unit length of the habit plane far away from the interface (in the $0-$th layer), then the $i-$th layer of the branched structure of the same out-of-plane thickness consists of $2^i/d_0$ branching segments, each of width $d_i=d_0/2^i$ and length $L_i$ as depicted in Fig.~\ref{hierarchical_KM}. Strictly speaking, this construction may not be literally \emph{self-similar}, as the $d_i/L_i$ ratio may change from layer to layer (cf. \citep{KM1,KM2,ChanConti1,ChanConti2}). However, as the term 
\emph{self-similar construction} has been introduced in \citep{KM1} and is commonly used for this type of construction, we will continue using it in this paper.

Let $L$ be the distance between the habit plane and the free surface of the crystal, as described previously. This distance can be expressed as 
\begin{equation}
 L=\sum_{i=0}^NL_i, \label{sumL}
\end{equation}
where $N$ is the total number of the branching generations. If $E_{\rm elast}^{(i)}$ and $E_{\rm surf}^{(i)}$ are, respectively, the elastic and surface energy of one segment in the $i-$th layer, then the total energy of the crystal is
\begin{equation}
 E=\frac{1}{d_0}\sum_{i=0}^{N}2^i\left(E_{\rm elast}^{(i)}+E_{\rm surf}^{(i)}\right)+d_0\frac{G_{AB}}{2^{N}}, \label{suma}
\end{equation}
where the last term represents the elastic energy localized at the habit plane.   As the number of the branching generations increases, this term goes quickly to zero since it is proportional to $2^{-N}$. Consequently, the energy of branched structures with several generation of branching ($N\gg{}1$) can be very accurately approximated just by the sum of the energy of the branching segments. According to the analyses by  \cite{KM1,KM2}, this energy scales with the length of the crystal as $E\sim{}{\sigma_{AB}^{2/3}}L^{1/3}$, and the energy-minimizing width of the twins $d(x)$ at the distance $x$ away from the austenite-martensite interface is $d(x)\sim{}x^{2/3}$.  

\subsection{The scaling argument} 
\label{scalingargument}
The simple constructions of microstructure at the austenite-martensite interface described above enabled \cite{KM1,KM2} to formulate the following fundamental \emph{scaling argument}: As the scaling of the total energy for the simple laminate is $E\sim{\sigma_{AB}^{1/2}L^{1/2}}$, while the scaling for the branched structure with $N\gg{1}$ is $E\sim{\sigma_{AB}^{2/3}L^{1/3}}$, the branching is always preferred for $L\rightarrow\infty$ or/and $\sigma_{AB}\rightarrow{}0$, regardless of the prefactors for the scalings.

Constructions by  \cite{Capella_Otto_1} and  \cite{ChanConti1,ChanConti2} with a more complete treatment of the surface energy revealed that for a large but finite number of branching generations the correct scaling is rather  $E\sim{}p_1{\sigma_{AB}^{2/3}L^{1/3}}+p_2\sigma_{AB}L$ (where $p_1$ and $p_2$ are properly chosen prefactors), which means that the scaling argument holds for $\sigma_{AB}\rightarrow{0}$ (with $L$ finite), but not for $L\rightarrow{}\infty$ (with $\sigma_{AB}$ finite). However, the limit $L\rightarrow{}\infty$ is not very interesting from the physical point of view, while $\sigma_{AB}\rightarrow{0}$ can be a realistic description for some alloys. For example, the so-called $a/b-$twins in the Ni-Mn-Ga shape memory alloy can have up to 10$^3$ times smaller surface energy than other twinning systems in the same alloy \citep{Zeleny} or twins in other alloys \citep{shilo,waitz}. This may be the origin of extensive coarsening of twins observed in the seven-layer modulated structure of Ni-Mn-Ga \citep{Kaufman}.

\section{The self-similar construction in a full three-dimensional setting}

In this section, we will follow the approach of  \cite{KM1,KM2} to construct an upper bound of the total energy for a branched microstructure given a pair of variants satisfying equations (\ref{compat1}) and (\ref{compat2}). First, we will propose a continuous displacement field providing a twin refinement towards the phase boundary. Then we will evaluate the energy of this displacement field and discuss the scaling laws. Unlike in references \citep{KM1,KM2,Capella_Otto_1,Capella_Otto_2,ChanConti1,ChanConti2}, we will not do the construction by prescribing directly the displacement field ${\mathbf y}({\mathbf x})$ over the branched structure; instead, we will find a piecewise-constant deformation gradient ${\mathbf F}$, such that the planar interfaces between the regions where ${\mathbf F}$ is constant satisfy the kinematic compatibility conditions. This will ensure the existence of a displacement field that is continuous with $\nabla{\mathbf y}={\mathbf F}$ almost everywhere. 
 
 Following the approach of Kohn and M\"uller, we will also assume that the twin interfaces are approximately parallel to the stress-free $AB$ twins (i.e., perpendicular to ${\mathbf n}$), and we will assume that the surface energy can be expressed as $\sigma_{AB}$ times the area of the interface. 

\subsection{Construction of the deformation gradients}

Consider now $\mathbf{A}$ and $\mathbf{B}$ satisfying the compatibility conditions (\ref{compat1},\ref{compat2}). The vectors ${\mathbf n}$ and ${\mathbf m}$ are now not required to be perpendicular to each other, and the volume fraction $0<\lambda<1$ is also general. Hence, our construction is applicable for any symmetry class of martensite, and for any lattice parameters such that the conditions (\ref{compat1},\ref{compat2}) are satisfied. 

We propose a branching segment (Fig.~\ref{branched}) consisting of five regions with homogeneous deformation gradients. These five gradients are denoted  as $\mathbf{A}^{+}_i$, $\mathbf{A}^{-}_i$,$\mathbf{B}^{+}_i$, $\mathbf{B}^{-}_i$ and $\mathbf{B}^{*}_i$ in Fig.~\ref{hierarchical_KM}. In the simplest case, we will assume that four of these deformation gradients lie exactly on the energy wells; in particular, we assume that  
\begin{equation}
\mathbf{B}^{+}_i = \mathbf{B}^{-}_i = \mathbf{B}^{*}_i = \mathbf{B}
\end{equation}
and
\begin{equation}
\mathbf{A}^{-}_i = \mathbf{A}
\end{equation}
for all $i$, while the remaining one is slightly elastically strained, 
\begin{equation}
\mathbf{A}^{+}_i =\mathbf{A} + \boldsymbol{\delta} \mathbf{A}_i.
\end{equation}
Due to this elastic strain, the interface between $\mathbf{A}^{+}_i$ and $\mathbf{B}^{+}_i = \mathbf{B}^{*}_i$ is inclined, with the new orientation being given by a unit vector ${\mathbf n}+\boldsymbol{\delta} {\mathbf n}_i$.  To realize a continuous deformation from this strain field, the perturbations $\boldsymbol{\delta} \mathbf{A}_i$ and $\boldsymbol{\delta} \mathbf{n}_i$ are constrained by additional compatibility conditions beyond (\ref{compat1}) and (\ref{compat2}).  
The compatibility conditions at the inclined planar interfaces are 
\begin{equation}
(\mathbf{A}+\boldsymbol{\delta}\mathbf{A}_i)-\mathbf{B}={\mathbf c}_i\otimes({\mathbf n}+ \boldsymbol{\delta}{\mathbf n}_i),\label{compat_delt_n}
\end{equation}
and the compatibility conditions for connecting the $i-$th layer with the neighboring layers are
\begin{equation}
(\mathbf{A}+\boldsymbol{\delta} \mathbf{A}_i)-\mathbf{A} = \boldsymbol{\delta} \mathbf{A}_i={\mathbf d}_i\otimes{\mathbf m} \label{compat_delt_A}
\end{equation}
and
\begin{equation}
(\mathbf{A}+\boldsymbol{\delta} \mathbf{A}_{i\pm{1}})-(\mathbf{A}+\boldsymbol{\delta}\mathbf{A}_i)=\boldsymbol{\delta}\mathbf{A}_{i\pm{}1}-\boldsymbol{\delta}\mathbf{A}_i={\mathbf d}_{i}^{\pm}\otimes{\mathbf m}
\label{compat_delt_pm}
\end{equation}
for some vectors ${\mathbf c}_i$, ${\mathbf d}_i$, ${\mathbf d}^+_i$, and ${\mathbf d}^-_i$.

The last condition (\ref{compat_delt_pm}) is, in fact, redundant.  Indeed, if (\ref{compat_delt_A}) is satisfied for each $i$, then
\begin{equation}
\boldsymbol{\delta} \mathbf{A}_{i\pm{}1}-\boldsymbol{\delta}\mathbf{A}_i=({\mathbf d}_{i\pm{1}}-{\mathbf d}_{i})\otimes{\mathbf m},
\end{equation}
which means that (\ref{compat_delt_pm}) is satisfied with 
\begin{equation} {\mathbf d}_i^{\pm}={\mathbf d}_{i\pm{1}}-{\mathbf d}_{i}. \end{equation}
Consequently, the existence of the vectors ${\mathbf d}_i$ directly implies the existence of the vectors ${\mathbf d}_i^{\pm}$.

\begin{figure}
\centering
 \includegraphics[width=0.5\textwidth]{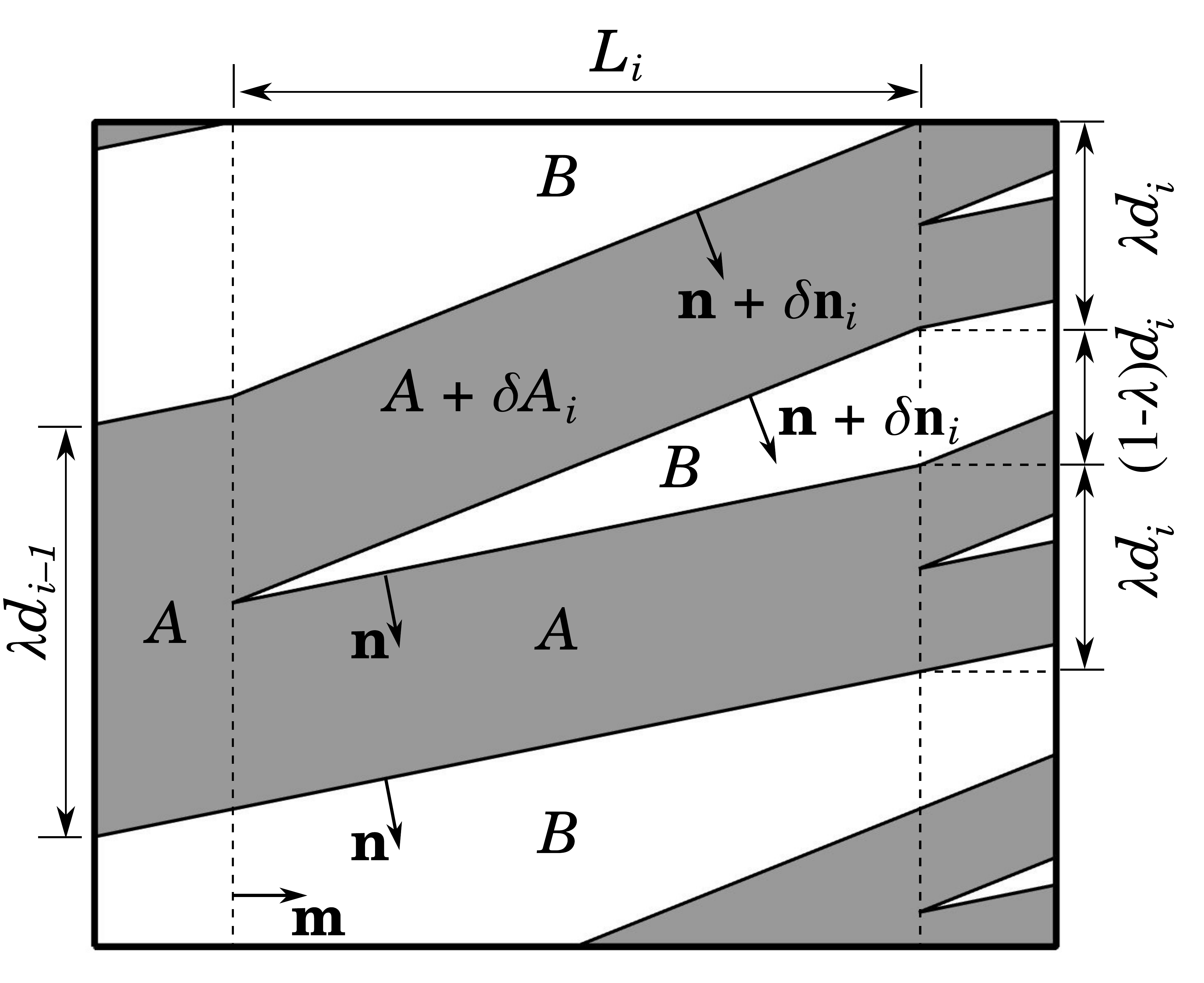}
 \caption{The proposed segment of the branched structure (segment in the $i-$th layer). White color denotes variant $\mathbf{B}$ (stress-free), darker gray corresponds to variant $\mathbf{A}$, which is, in some regions, slightly elastically strained ($\mathbf{A}+\boldsymbol{\delta}\mathbf{A}_i$). Notice that the sketch shows only one of two alternating connections of the branching segment to the $(i-1)-$th layer. In the second possible case, $\mathbf{A}$ and $\mathbf{A}+\boldsymbol{\delta}\mathbf{A}_i$ are connected to $\mathbf{A}+\boldsymbol{\delta}\mathbf{A}_{i-1}$ instead of to $\mathbf{A}$. This alternation is reflected by equations (\ref{compat_delt_A}) and (\ref{compat_delt_pm}).}\label{branched}
\end{figure}

Our aim then is to find a perturbation $\boldsymbol{\delta} \mathbf{A}_i$ such that these compatibility conditions are satisfied for prescribed $\boldsymbol{\delta} {\mathbf n}_i$ (i.e., for a prescribed inclination of the interface). As the vectors ${\mathbf n}$, ${\mathbf m}$ and $({\mathbf n}+\boldsymbol{\delta}{\mathbf n}_i)$ are necessarily coplanar, but ${\mathbf n}$ and ${\mathbf m}$ are never collinear, there always exists a scalar parameter $\varepsilon_i$ such that
\begin{equation}
 {\mathbf n}+\boldsymbol{\delta}{\mathbf n}_i = \frac{{{\mathbf n}+\varepsilon_i{\mathbf m}}}{\left|{{\mathbf n}+\varepsilon_i{\mathbf m}}\right|}.
\end{equation}
The parameter $\varepsilon_i$ has also a direct geometrical interpretation. It can be easily shown that 
\begin{equation}
 \varepsilon_i=\frac{(1-\lambda)\sqrt{1-({\mathbf m}\cdot{\mathbf n})^2}d_i}{L_i}\stackrel{\rm def.}{=}\frac{(1-\lambda)\alpha{}d_i}{L_i}, \label{diLi}
\end{equation}
i.e., that $\varepsilon_i$ determines the ratio between $d_i$ and $L_i$.

To satisfy the condition (\ref{compat_delt_n}), it is then sufficient to take
\begin{equation}
 {\mathbf d}_i=\varepsilon_i{\mathbf a}.
\end{equation}
Indeed, utilizing (\ref{compat1}),
\begin{equation}
(\mathbf{A}+\boldsymbol{\delta}\mathbf{A}_i)-\mathbf{B} = (\mathbf{A}-\mathbf{B})+\boldsymbol{\delta}\mathbf{A}_i = {\mathbf a}\otimes{\mathbf n}+\varepsilon_i{\mathbf a}\otimes{\mathbf m}={\mathbf a}\otimes({{\mathbf n}+\varepsilon_i{\mathbf m}}),
\end{equation}
which is a rank-one matrix. By taking 
\begin{equation}
 {\mathbf c}_i = \left|{{\mathbf n}+\varepsilon_i{\mathbf m}}\right|{\mathbf a},
\end{equation}
we obtain exactly (\ref{compat_delt_n}). 

In summary, we have shown that the branching segment sketched in Fig.~\ref{branched} represents a continuous displacement field if the small perturbation of the deformation gradient in one of the regions is
\begin{equation}
 \boldsymbol{\delta} \mathbf{A}_i=\varepsilon_i{\mathbf a}\otimes{\mathbf m}. \label{deltaA}
\end{equation}
Then, the resulting inclination of the twinning planes encapsulating the elastically strained region is
\begin{equation}
\boldsymbol{\delta} {\mathbf n}_i = \frac{{{\mathbf n}+\varepsilon_i{\mathbf m}}}{\left|{{\mathbf n}+\varepsilon_i{\mathbf m}}\right|}-{\mathbf n},
\end{equation}
and the $d_i/L_i$ ratio is given by (\ref{diLi}). The segments can be used to construct a fully compatible, three-dimensional branched structure, as each layer of the structure inherently satisfies a macroscopic compatibility condition with austenite. In particular, the macro-scale deformation gradient in the $i-$th layer of the structure is
\begin{equation}
\lambda{}\mathbf{A}+(1-\lambda)\mathbf{B}+\frac{\lambda}{2}\boldsymbol{\delta} \mathbf{A}_i={\mathbf I}+{\mathbf b}\otimes{\mathbf m}+\frac{\lambda}{2}\varepsilon_i{\mathbf a}\otimes{\mathbf m}={\mathbf I}+({\mathbf b}+\frac{\lambda}{2}\varepsilon_i{\mathbf a})\otimes{\mathbf m},
\end{equation}
which is compatible with austenite over a planar interface perpendicular to ${\mathbf m}$.

Before discussing the energy of the proposed construction, let us mention that the branching segment in Fig.~\ref{branched} is very similar to the real geometry of the branching points observed in shape memory alloys. In Fig.~\ref{experiment}, two examples of such observations are seen. Fig.~\ref{experiment}(a) shows a branching point in a Type-II laminate in a Cu-Al-Ni single crystal observed by white-light interferometry (see \citep{seiner} for more details on the experiment), and Fig.~\ref{experiment}(b) shows a branching microstructure in ten-layer modulated (10 M) martensite in a Ni-Mn-Ga single crystal \citep{Bronstein}. In both cases, the branching appears to be provided by planar interfaces, with (approximately) one half of the original twin band of the minor variant remaining straight and the second half becoming tilted beyond the branching point.

\begin{figure}[htb]
\centering
 \includegraphics[width=0.8\textwidth]{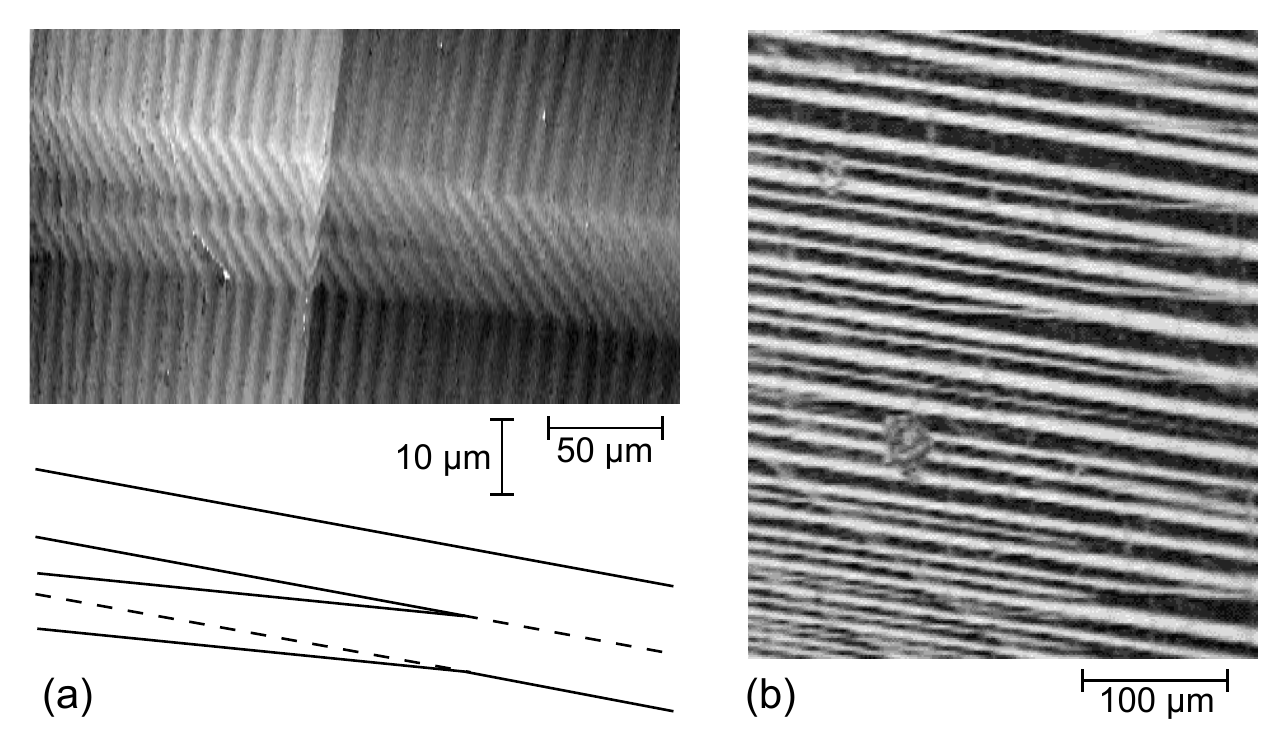}
 \caption{Experimental observations of the morphology of the branching points: (a) White-light interferometry image of a single branching point in a Cu-Al-Ni single crystal; the orientation of the fine interference fringes correspond to the tilt of the observed surface, and the twin interfaces are seen as sharp changes of this orientation. The arrangement of the twin interfaces read from the micrograph is shown below the image; the dashed lines are parallel to the upper interface. (b) Optical micrograph of several branching points in a Ni-Mn-Ga single crystal (courtesy of E. Bronstein; see \citep{Bronstein} for more details on the experiment.)}\label{experiment}
\end{figure}

\subsection{Energy considerations}
The energy of the  microstructure consists of the elastic energy of the branching segments, the surface energy of the twin walls, and the energy of the closure domains ($G_{AB}$, see (\ref{laminate}) and ({\ref{suma}})). For the self-similar construction, we will assume that the energy of the closure domains is negligible (a similar result is obtained if this energy is comparable with the energy of the last layer of the branched structure \citep{ChanConti1,ChanConti2}); this assumption will be discussed in more detail in Section 4. The elastic energy of the branching segments is simple to express. In a given segment, the elastic energy is located only in the region with the deformation gradient $\mathbf{A}+\boldsymbol{\delta} \mathbf{A}_i= (\mathbf{I} + \boldsymbol{\delta}\mathbf{A}_i \mathbf{A}^{-1}) \mathbf{A}$. The elastic energy density in this region is
\begin{equation}\label{varphiE}
 \varphi_E = \frac{1}{2}\left(\sym \boldsymbol{\delta}\mathbf{A}_i \mathbf{A}^{-1}\right) \colon{}{\mathbb C}\colon{}\left( \sym \boldsymbol{\delta} \mathbf{A}_i \mathbf{A}^{-1} \right),
\end{equation}
where ${\mathbb C}$ is a tensor of elastic constants for the martensite phase (with major and minor symmetry),  and $\sym \mathbf{B} \stackrel{\rm def.}{=} \frac{1}{2}(\mathbf{B} + \mathbf{B}^T)$ gives the symmetric part of a tensor $\mathbf{B}$.
Using (\ref{deltaA}), this expression simplifies to 
\begin{equation}
 \varphi_E = \frac{1}{2}\varepsilon_i^2(\sym {\mathbf a}\otimes{\mathbf{A}^{-T} \mathbf m})\colon{}{\mathbb C}\colon{}(\sym {\mathbf a}\otimes{\mathbf{A}^{-T} \mathbf m})\stackrel{\rm def.}{=}\frac{1}{2}\mathcal{C} \varepsilon_i^2,
\end{equation}
where $\mathcal{C}$ collects the dependence on elastic constants and the geometry of the twins through the vectors ${\mathbf a}$ and ${\mathbf m}$ and the transformation tensor $\mathbf{A}$. The elastic energy of one segment is then 
\begin{equation}
 E^{(i)}_{\rm elast.} = \lambda \frac{L_id_i}{2}\mathcal{C}\varepsilon_i^2=\lambda\frac{(1-\lambda)\alpha{}d_i^2}{2}\mathcal{C}\varepsilon_i.
\end{equation}
As the $i-$th layer consists of $1/d_i = 2^i/d_0$ segments, the total elastic energy of the branched structure per unit length of the habit plane is  
\begin{equation}
 E_{\rm elast.} =\frac{1}{d_0}\sum_{i=0}^{N}2^iE^{(i)}_{\rm elast.}=\lambda(1-\lambda)\mathcal{C}\alpha{}d_0 \sum_{i=0}^N\frac{\varepsilon_i}{2^{i+1}}. \label{Eelast}
\end{equation}

Provided that the inclinations are very small ($\varepsilon_i\ll{}1$, $d_i{}\ll{}L_i$), we can approximately assume that $\sigma_{AB}$ is a constant and that the twin interfaces are all parallel to the stress-free twinning plane. Then, the surface energy of one segment is
\begin{equation}
 E^{(i)}_{\rm surf.}=4\sigma_{AB}\frac{L_i}{\alpha}=4\sigma_{AB}(1-\lambda)\frac{d_i}{\varepsilon_i}, \label{surf_i}
\end{equation}
and the total surface energy in the branched structure per unit length of the habit plane is
\begin{equation}
 E_{\rm surf.}=4\sigma_{AB}(1-\lambda)\sum_{i=0}^{N}\frac{1}{\varepsilon_i}.
\end{equation}
This sum converges for $N\rightarrow\infty$ only if
\begin{equation}
\lim_{i\rightarrow\infty}\varepsilon_i=\infty. 
\end{equation}
Then, however, as $N$ becomes very large, the assumption that $d_i{}\ll{}L_i$ is violated, and the surface energy term cannot be expressed by (\ref{surf_i}).  Nevertheless, if $d_0{}\ll{}L_0$, and if $\varepsilon_i$ does not grow too fast, the approximation (\ref{surf_i}) is sufficiently justified. For the purpose of the construction in this section, we will assume that $N$ is large enough to make the energy of the closure domains negligible, while ensuring $d_i{}\ll{}L_i$ remains fulfilled. Such an assumption allows us to follow the original construction of Kohn and M\"{u}ller; the other cases will be discussed in Section 4.  

The total energy of the branched structure can be expressed as
\begin{equation}
 E=(1-\lambda)\sum_{i=0}^N\left[\lambda{}\mathcal{C}\alpha{}d_0\frac{\varepsilon_i}{2^{i+1}}+\frac{4\sigma_{AB}}{\varepsilon_i}\right] \label{simplified}
\end{equation}
and can be minimized with respect to $\varepsilon_i$ for fixed $d_0$, which is equivalent to minimizing with respect to $1/d_0$ for fixed $L$ as done in \citep{KM1,KM2,ChanConti1,ChanConti2}.
Since $\varepsilon_i$ appears only in the $i-$th term, minimization with respect to this parameter is relatively simple and the minimum is reached for 
\begin{equation}
 \varepsilon_i = \sqrt{\frac{8\sigma_{AB}}{\lambda{}\mathcal{C}\alpha{}d_0}}{(\sqrt{2})^i}.\label{epsopt}
\end{equation}
As expected, $\varepsilon_i$ rapidly increases to reduce the growth of the surface energy (\ref{surf_i}). After substituting (\ref{epsopt}) into (\ref{simplified}), the total energy is
\begin{equation}
 E=(1-\lambda)\sqrt{2\sigma_{AB}\lambda{}\mathcal{C}\alpha{}d_0}\sum_{i=0}^{N}\frac{1}{(\sqrt{2})^{i}}\stackrel{\rm def.}{=}\left[(1-\lambda)\sqrt{2\sigma_{AB}\lambda{}\mathcal{C}\alpha{}d_0}\right]a_N,
 \label{simplified_final}
\end{equation}
and the length of the martensite part of the crystal is 
\begin{equation*}
 L=\sum_{i=0}^N{L_i}=\sum_{i=0}^N\frac{(1-\lambda)\alpha{}d_0}{2^i\varepsilon_i}=
 (1-\lambda)\alpha^{3/2}{d_0}^{3/2}\sqrt{\frac{\lambda{}\mathcal{C}}{8\sigma_{AB}}}{\sum_{i=0}^N\frac{1}{(2\sqrt{2})^i}} 
\end{equation*}
\begin{equation}
 \stackrel{\rm def.}{=}\left[(1-\lambda)\alpha^{3/2}{d_0}^{3/2}\sqrt{\frac{\lambda{}\mathcal{C}}{8\sigma_{AB}}}\right]b_N,
 \label{length}
\end{equation}
where we introduced partial sums\begin{equation}
 a_N=\sum_{i=0}^N{\frac{1}{(\sqrt{2})^i}} {\mbox{\hspace{5mm} and \hspace{5mm}}} b_N=\sum_{i=0}^N{\frac{1}{(2\sqrt{2})^i}}.
\end{equation}
Regardless of the number of the branching generations, (\ref{length}) clearly gives the scaling $d_0\sim{L}^{2/3}$ predicted by  \cite{KM1,KM2}. According to (\ref{simplified}), the energy scales as $E\sim\sqrt{d_0}$. Consequently, the scaling of the energy with respect to the length is $E\sim{L}^{1/3}$, which is again in agreement with \citep{KM1,KM2}. Finally, by expressing $d_0$ from (\ref{length}) and substituting it into (\ref{simplified_final}), we can confirm that  
\begin{equation}
 E\sim{\sigma_{AB}^{2/3}\mathcal{C}^{1/3}L^{1/3}}\label{scalingE}
\end{equation}
as also predicted in \citep{KM1,KM2}. The explicit formulas for $d_0$ and $E$ are
\begin{equation}
 d_0=\left({\frac{8\sigma_{AB}}{\lambda{}\mathcal{C}}}\right)^{1/3}\frac{L^{2/3}}{\alpha{}(1-\lambda)^{2/3}}b_N^{-2/3} \label{DKM}
\end{equation}
and 
\begin{equation}
 E=2\left[\lambda{}(1-\lambda)^2\right]^{1/3}\sigma_{AB}^{2/3}\mathcal{C}^{1/3}L^{1/3}a_Nb_N^{-1/3}. \label{EKM}
\end{equation}

\subsection{Generalizations of the construction}\label{generalization}
In this subsection, we will propose three modifications of the construction outlined above, and discuss which of these modifications leads to a decrease of the energy for the construction, i.e., an improvement on the upper bound. The aim of this discussion is to show that our construction is sufficiently versatile to capture various effects that may appear in real branched microstructures.

\begin{figure}
 \includegraphics[width=\textwidth]{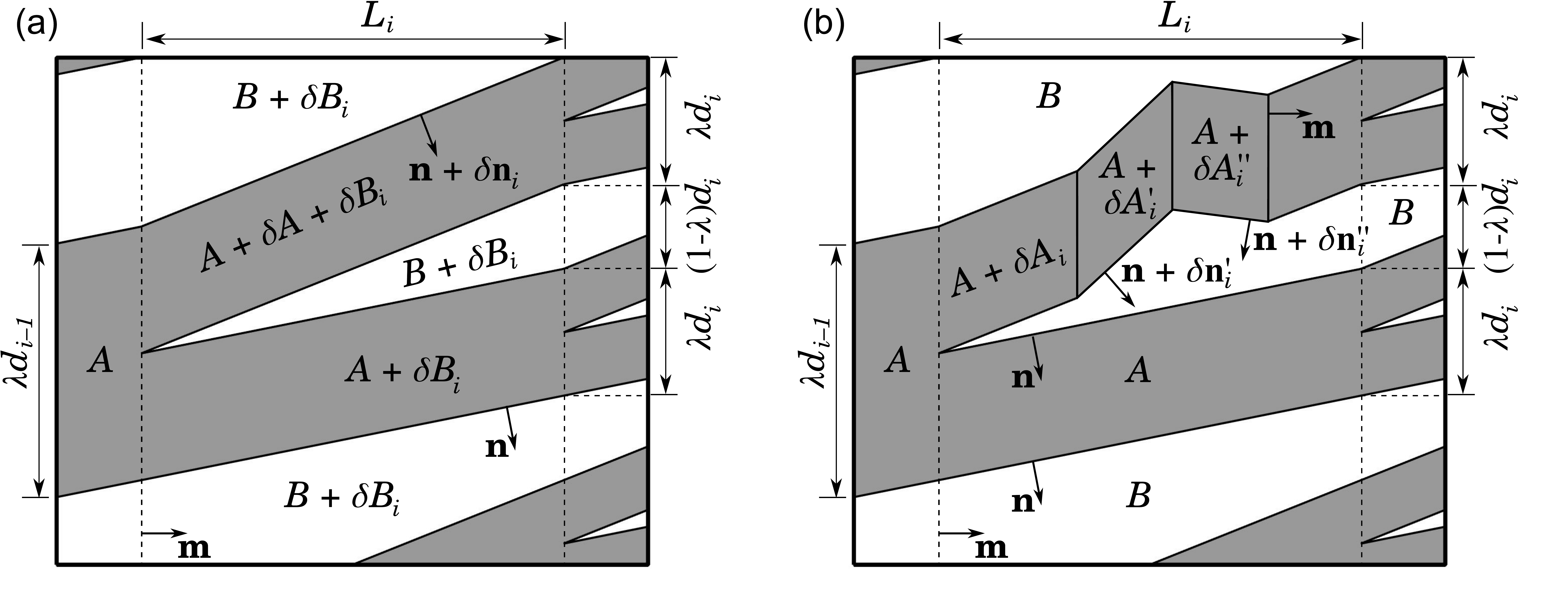}
 \caption{Modifications of the construction of the branching segment: (a) delocalization of the elastic energy; (b) curved interface with $\boldsymbol{\delta}\mathbf{A}_i'$ and $\boldsymbol{\delta}\mathbf{A}_i''$ being the elastic strains in different parts of the layer of variant $\mathbf{A}$, and $\boldsymbol{\delta}{\mathbf n}_i'$ and $\boldsymbol{\delta}{\mathbf n}_i''$ being the respective inclinations of the twinning planes.}\label{reduction}
\end{figure}

\begin{itemize}
 \item {\bfseries Delocalization of the elastic energy in the branching segment -- } One of the obvious unrealistic assumptions of the above construction is that the elastic strains are localized only in the $\mathbf{A}^{+}_i$ region while the rest of the branching segment is strain-free. For example, if we take a homogeneous deformation gradient $\boldsymbol{\delta}\mathbf{B}_i$ and assume that the deformation gradients in the branching segment are  (see Fig.~\ref{reduction}(a))
 \begin{equation}
  \mathbf{A}_i^{+}=\mathbf{A}+\boldsymbol{\delta}\mathbf{A}_i+\boldsymbol{\delta}\mathbf{B}_i,
 \end{equation}
 \begin{equation}
  \mathbf{A}_i^{-}=\mathbf{A}+\boldsymbol{\delta}\mathbf{B}_i,
 \end{equation}
 and 
 \begin{equation}
 \mathbf{B}_i^{-}=\mathbf{B}_i^{+}=\mathbf{B}_i^{*}=\mathbf{B}+\boldsymbol{\delta}\mathbf{B}_i,
 \end{equation}
the orientations of the twinning planes inside the branching segments (i.e., the geometry of the segment) remains unchanged. If, furthermore, 
\begin{equation}
\boldsymbol{\delta}\mathbf{B}_i=\delta_i{\mathbf a}\otimes{\mathbf m},
 \end{equation}
where $\delta_i$ is some scalar parameter, also the compatibility with the $(i-1)-$th layer and the $(i+1)-$th layer remains unbroken. The elastic energy of branched structure is minimal if
\begin{equation}
 \delta_i=-\frac{\lambda}{2}\varepsilon_i,
\end{equation}
and the total energy (\ref{EKM}) is then reduced $(1-\frac{\lambda}{2})^{1/3}$ times. 

As $\delta_i$ is proportional to $\varepsilon_i$, and $\varepsilon_i$ increases as $(\sqrt{2})^i$, also the homogeneous shear strains must increase in the branched structure in the vicinity of the interface. This may be the reason why the branching laminates in Fig.~\ref{branchingEx} appear slightly curved when approaching the habit plane. This, however, does not mean that the interfaces are curved inside of the branching segments (as discussed in the next point), rather it implies that the planar interfaces in the individual segments are getting more and more inclined in the deformed configuration, as the shear strains increase.

An even more significant delocalization of the elastic energy appears if we allow all the interfaces to tilt. As shown in the detailed construction given in the Appendix and discussed in Section 4, such a delocalization may lead to a significant reduction of the elastic energy. However, all these modifications of the branching segment just reduce the total energy of the branched structure by a scalar prefactor, i.e., they do not affect the scaling laws. It is also interesting to note that the tilting of all interfaces does not seem to occur for the experimentally observed branching, 
Fig.~\ref{experiment}. Possibly, this aspect of  real branched microstructures results not only from requirements of energy minimization, but also from requirements of kinematics and energy dissipation (cf., \citep{seiner}). Such a discussion, however, falls beyond the scope of this paper.

\item {\bfseries  Curved interfaces -- }Another obvious simplification of our construction is that we assume planar interfaces. However, as $\boldsymbol{\delta}\mathbf{A}_i=\varepsilon_i{\mathbf a}\otimes{\mathbf m}$, the parameter $\varepsilon_i$ can vary spatially within each segment, provided that the resulting strain field with this variation is a gradient field and that the connection to the $(i+1)-$th layer at the right-hand-side edge of the segment remains unchanged, as sketched in Fig.~\ref{reduction}(b). For example, it is possible to consider the variation $\varepsilon_i(x)=\varepsilon^0_i + \delta\varepsilon_i(x)$ for $0\leq{}x\leq{}L_i$.  (Indeed, notice that this variation is the gradient of the map $[\int_0^{\mathbf{x} \cdot \mathbf{m}} (\varepsilon_i(x) + \varepsilon_i^0) {\rm d}x ]\mathbf{a}$.)
The elastic energy of the branching segment is then 
\begin{equation}
E^{(i)}_{\rm elast.}=\frac{\lambda{}d_i}{2}
 \mathcal{C}\int_0^{L_i} \left( \varepsilon^0_i+\delta{}\varepsilon_i(x)\right)^2{\rm d}x.
\end{equation}
However, due to the condition
\begin{equation}
\int_0^{L_i}\delta{}\varepsilon_i(x){\rm d}x=0,
\end{equation}
it can be easily shown that the elastic energy is minimal for $\delta{}\varepsilon_i(x)\equiv{0}$ for all $x$.  The surface energy is also minimal for $\delta{}\varepsilon_i(x)\equiv{0}$, since the planar interface has the smallest area. Hence, we can conclude that the energy (\ref{simplified}) of the proposed construction is minimal for  planar interfaces. This is in a good agreement with experimental observations (Fig.~\ref{experiment}), where the interfaces appear to be tilted but not curved. 

\item {\bfseries  Volume fraction variations -- } Since the total energy of the branched structure depends on $\lambda$ as $\left[\lambda{(1-\lambda)^2}\right]^{1/3}$, it is obvious that this energy is very sensitive to small variations of $\lambda$, especially for $\lambda$ close to 0 or 1, where the derivative $\partial{E}/\partial{\lambda}$ goes to $\pm{}\infty$. This means that the energy release due to making $\lambda$ slightly closer to 0 or 1 may be much larger than the elastic energy to be paid for violating the the macro-scale compatibility condition (\ref{compat2}). Using a similar construction as in Subsection 3.1, one can show that compatibility for perturbed $\lambda$ can be achieved when the elastic energy density in the branched structure is increased by a term proportional to $(\delta\lambda)^2$. Due to geometry reasons, this energy increase cannot be localized just in the vicinity of the habit plane. Instead, it must be spread over the whole martensite region. In other words, the variation of the volume fraction may be beneficial only for small $L$.

\end{itemize}

\section{Discussion beyond the self-similar concept}

\subsection{Motivation}

The construction presented above gives a simplified approximation of the real branched microstructure and was done under several assumptions that need to be discussed in more detail. In this section, we comment on these assumptions and propose a modification of the construction to make the model mimic real microstructures more accurately.

The first questionable point in the self-similar construction is that it does not take into account the specific conditions in the  $0-$th layer. This layer does not need to have the same topology as the following layers, as it is not connected to any preceding layer. Hence, as also noted by \citep{ChanConti1}, it can consist of just a simple laminate with the width of the twins equal to $d_0$ 
(Fig.~\ref{L0_fig}). By realizing that the $0-$th layer has no reason for elastic strains, the total energy is significantly reduced. 

\begin{figure}
\centering
 \includegraphics[width=0.5\textwidth]{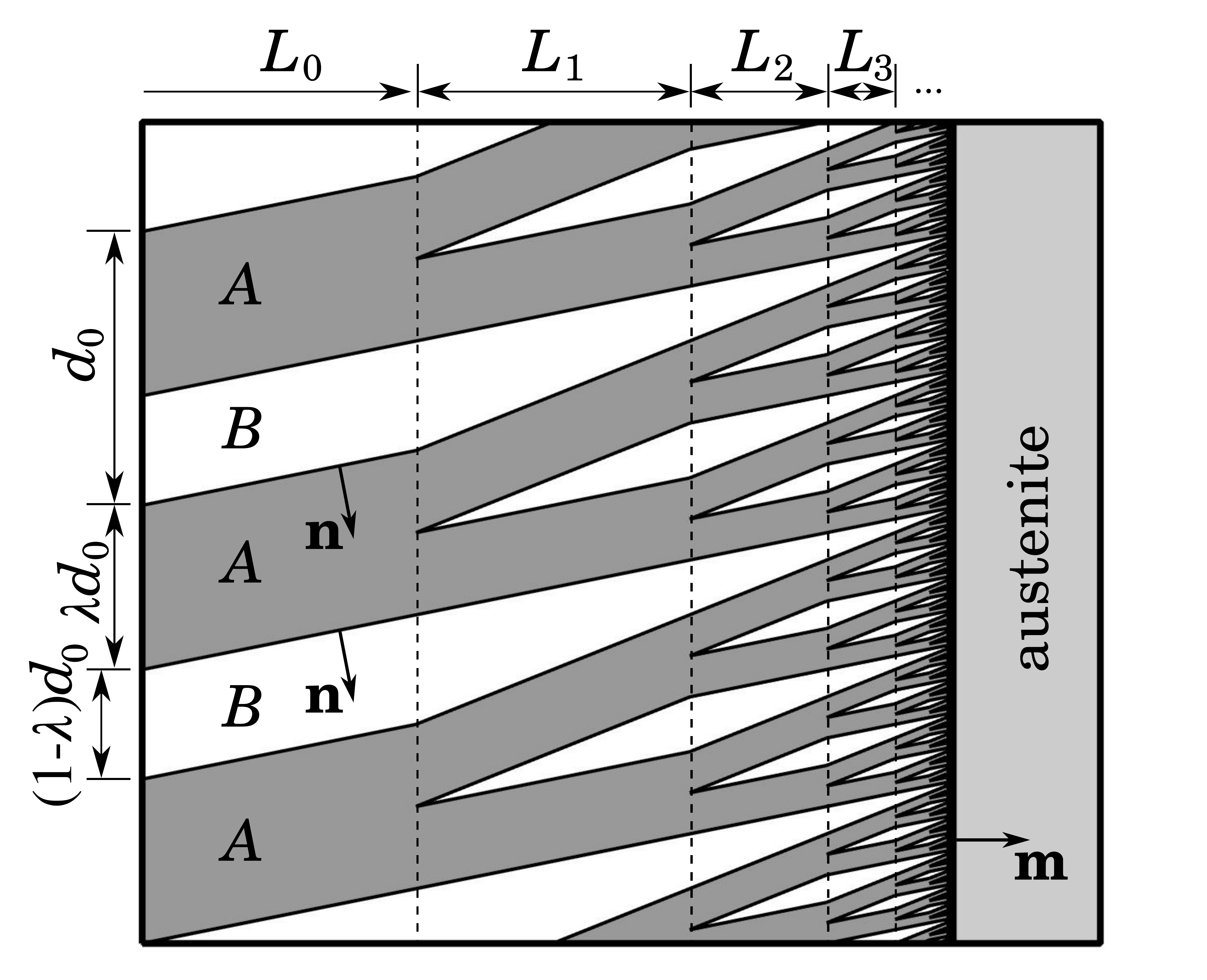}
 \caption{The branched structure after releasing the elastic energy from the $0-$th layer.}\label{L0_fig}
\end{figure}

The second questionable point relates to the number of the branching generations. As mentioned in Section 3.2, the scaling laws (\ref{DKM},\ref{EKM}) and the simple energy balance for branching microstructure (\ref{simplified}) are only meaningful for a certain range of branching generations $N$.  This range is bounded below by the minimal number of generations needed to justify neglecting the energy of the closure domains, and bounded above by the maximal number of generations for which the inclinations are small (i.e., each $\varepsilon_i < \ldots < \varepsilon_N$ should be $\ll 1$ so that the surface energy (\ref{surf_i}) is justified).   Whether the latter is larger than the former, so that there is a range of validity for this energy balance, depends on the detailed parameters. Let us notice that the $\varepsilon_i$ increases as $(\sqrt{2})^i$, which means by three orders of magnitude for 20 generations of branching, so the upper bound of the range can be very restrictive. 

At the same time, the number $N$ is not an a priori known material parameter, and so verifying that it falls into the certain range for the given set of the material parameters is obviously a very questionable approach.  \cite{ChanConti1,ChanConti2} solved this problem in an elegant manner by considering the branching to stop once $d_i/L_i\approx{}1$, i.e., when the interfaces are inclined by a certain angle\footnote{This condition can be modified to $d_i/L_i\approx{}c$, where $c$ is a small constant, to keep the assumption $\varepsilon_i\ll{}1$ satisfied. Such a modification does not qualitatively affect the result of \cite{ChanConti1,ChanConti2}.}. Then, the energy of the closure domains can be, up to a scalar factor, absorbed into the sum of energies of the branched structure.

From a more physical point of view, however, one expects the number of generations to arise from an energy balance comparing the branched structure to that of closure domains, and not from a purely geometrical condition (like the one put forth in \citep{ChanConti1,ChanConti2}). If $\varepsilon_i$ stays reasonably small, $E_{\rm elast.}+E_{\rm surf.}$ increases with the increasing number of branching generations, while the energy of the closure domains decreases. Apparently, there might be an optimum number of branching generations, and optimum twin width $d_0$ such that the total energy is minimal. Nevertheless, this minimum may not arise in the region of validity of the approximation (\ref{surf_i}). As a result, a more detailed model with an improved surface energy term is needed. 

Let us notice that the need for finite $N$ and the limited validity of (\ref{surf_i}) is not particular to any specific material parameters. It is geometrically impossible to construct an infinitely fine branching structure ($N\rightarrow\infty$) that does not cause the total length of interfaces in the structure to grow to infinity\footnote{Indeed, the total length of interfaces is always greater or equal than $\sum_{i=1}^N{L_i/d_i}$, and this sum is divergent unless $L_i/d_i\rightarrow{}0$. As $L_i/d_i\rightarrow{}0$, however, the interfaces between $\mathbf{A}+\boldsymbol{\delta} \mathbf{A}_i$ and $\mathbf{B}$ become more and more parallel to the habit plane, the length of the interfaces in one branching segment of the $i-$th layer becomes proportional to $d_i$, and the total sum of the length of interfaces grows to infinity.}. Due to (\ref{compat1}) and (\ref{deltaA}) the strain jump between $\mathbf{A}+\boldsymbol{\delta}\mathbf{A}_i$ and ${\mathbf B}$ is equal to ${\mathbf a}\otimes({\mathbf n}+\varepsilon_i{\mathbf m})$, and so it cannot diminish
 to zero for any $\varepsilon_i$, i.e., the interfaces are never smoothed out, which means there is always some non-zero surface energy associated with the interface. Hence, due to the length of interfaces going to infinity, the surface energy must diverge in the limit $N\rightarrow\infty$. As a consequence of this, there indeed should exist a minimum of energy corresponding to a specific number of the branching generations.

 \subsection{The extended model: main properties and a numerical test}
 
 The above motivation suggests that a more realistic construction of the branched structure should take into account releasing the elastic energy from the $0-$th layer, and the number of branching generations $N$ in this construction should be finite and should follow from energy minimization. To enable a realistic determination of $N$ from energy minimization, the current construction must be extended in two respects. First, we need to propose an upper bound estimate for the energy of the closure domains $G_{AB}$ depending on the same material parameters as used for calculating energy of the branched region. Second, we need to capture the increase of the surface energy in the branched structure in the limit $N\rightarrow\infty$, i.e., with increasing $\varepsilon_i$. With these two extensions, and  allowing the $0-$th layer to be unbranched, the construction becomes involved, and gaining any direct analytical insight into the properties of the energy minimizers of the construction becomes difficult. Nevertheless, several general conclusions can still be drawn, as shown in detail in the Appendix, where the extended construction is done under some simplifying assumptions. Moreover, under these assumptions, lower and upper bounds on the energy can be shown, and they collapse to the expected scaling law.  This is also discussed in the Appendix.  The assumptions are:
 \begin{enumerate}
  \item The elastic strains at the interface appear only in the martensite part of the crystal, and the elasticity is fully described by one isotropic elastic constant, which is the shear modulus $\mu$. The multi-well energy density $\varphi(\nabla{\mathbf y})$ is then assumed to be composed of two isotropic quadratic energy wells corresponding to the variants ${\mathbf A}$ and ${\mathbf B}$, which leads the form (\ref{assump2}). This elastic energy density is used for calculating both the elastic energy of the branched structure and of the closure domains.
  \item The surface energy terms are given by (\ref{surfE}); the constant $\sigma$ is assumed as universal, i.e. it applies both for the interfaces between the variants ${\mathbf A}$ and ${\mathbf B}$ and for the interfaces between these variants and the closure domains.
  \item The elastic energy in the branching segment is delocalized in a specific way (Fig.~\ref{BestAnsatz}) such that the strain jumps across all interfaces are the same. This delocalization leads to tilt of all interfaces and to heterogeneous strains in the region ${\mathbf B}_i^*$. Notice that the elastic energy of this branching segment is always smaller than the elastic energy of the original segment in Fig.~\ref{branched}; so the construction based on this segment leads to a better upper bound. 
 \end{enumerate}
 We will leave the details on this construction to the Appendix. Here  we will demonstrate that this extended model gives realistic predictions of length-scales and morphologies in real shape memory materials. The numerical procedure described in Subsection \ref{ssec:numProc} determines, for each given set of the most basic material parameters (${\mathbf A}$, ${\mathbf B}$, $\mu$ and $\sigma_{AB}=\sigma{}|{\mathbf a}|$) and the length of the martensite part of the crystal $L$, a set of optimized parameters $d_0^*$ (twin width in the $0-$th layer), $L_0^*$ (length of the $0-$th layer) and $N^*$ (the energy-minimizing number of branching generations).

For the numerical test, we study the branched microstructure shown in the optical micrograph in Fig.~\ref{branchingEx}.  The sample is a Cu-Al-Ni alloy (Cu$_{69}$Al$_{27.5}$Ni$_{3.5}$) undergoing the cubic-to-orthorhombic phase transformation, and the observed microstructure is a habit plane between austenite and a laminate of Type II twins. 
 
 From the lattice parameters of this alloy and the corresponding Bain matrices \citep{Bhattacharya}, the input parameters for the numerical procedure were calculated as 
 \begin{equation}\label{crys2}
\mathbf{A} = \left(
\begin{array}{ccc}
 1.0412 & 0.0468 & -0.0204 \\
 -0.0541 & 0.9154 & 0.0538 \\
 -0.0159 & -0.0464& 1.0411 \\
\end{array}
\right), \quad \mathbf{B} = \left(
\begin{array}{ccc}
 0.9150 & -0.0793 & 0.0215 \\
 0.0689 & 1.0385 & 0.0129 \\
 -0.0198 & -0.0503 & 1.0424 \\
\end{array}
\right),
\end{equation}
which resulted in
\begin{equation}\label{crys1}
\mathbf{n} = \left(\begin{array}{c} 0.6884 \\ 0.6884 \\ -0.2286 \end{array}\right), \quad \mathbf{m} =  \left(\begin{array}{c} 0.2611  \\ 0.7274 \\ -0.6346 \end{array}\right), \quad \mathbf{a} = \left(\begin{array}{c}  0.1833 \\  -0.1788 \\ 0.0057  \end{array}\right), \quad \lambda = 0.6992, \quad \alpha = 0.5642.
\end{equation}
From the angle between the traces of the habit plane and the twining planes observed in the micrograph (approximately 53$^\circ$), the cut plane was identified as $\mathbf{p} =( 0.0404, 0.9986, -0.0340)$ in the reference configuration. This is in a good agreement with the assumption that the crystal was cut approximately along the principal planes.

For the material parameters we note that the shear modulus in Cu-Al-Ni is strongly anisotropic and varies from 10 GPa to 100 GPa depending on the loading direction. As a realistic estimate for the characteristic shear modulus of martensite, we take $\mu = 70$ GPa.
 The interfacial parameter $\sigma_{AB} = \sigma | \mathbf{a}|$ is less certain: published results include $70$ mJ.m$^{-2}$  (\citep{shilo}, Cu-Al-Ni, curvature measurement), $187$ mJ.m$^{-2}$ (\citep{waitz}, Ni-Ti, first-principle calculation) and $530$ mJ.m$^{-2}$ (\citep{seiner}, Cu-Al-Ni microstructure scaling).  In addition, the total length of the martensite band can vary from a few microns to several millimeters.   Thus,  we take $\mu = 70$ GPa, $\sigma_{AB} = 100$ mJ.m$^{-2}$ and $L = 5$ mm for an explicit calculation of the laminate microstructure. For these material parameters, the optimized energy of the configurations $E^{(N)}$ with $N = 0,1,\ldots, 25$ branching generations was calculated with the result provided in Fig.~\ref{fig:Sigma200mu70}(a).   This gives $N^{\star} = 11$ as the optimal number branching generations.  The energy per unit depth and width for this energy minimizing configuration is $E^{\star} = E^{(N^{\star})} = 52.6$ J/m$^2$, the ratio of the unbranched length to total length is $L_0^{\star}/L  = 0.66$, and the twin width of the unbranched segment is $d_0^{\star} = 101.1$ $\mu$m. Notice that, from the energy $E^{\star}$, only approximately 3 J/m$^2$ is stored in the surface energy of the twins of the unbranched $0-$th layer. The rest is the energy of the branched structure, which means that the energy per unit area of a habit plane is more than two orders of magnitude higher than the surface energy of the twins. However, the branching is obviously strongly energetically preferred over the simple laminate microstructure, as the $E^{(N)}$ sharply decreases from $0$ to $N^\star$ in Fig.~\ref{fig:Sigma200mu70}(a).  
 Finally, a direct one-to-one comparison of the numerical model for these parameters and the experiment is given in Fig.~\ref{fig:Sigma200mu70}(b-c), and this is in striking agreement.

\begin{figure}
\centering
\includegraphics[width = 6.2in]{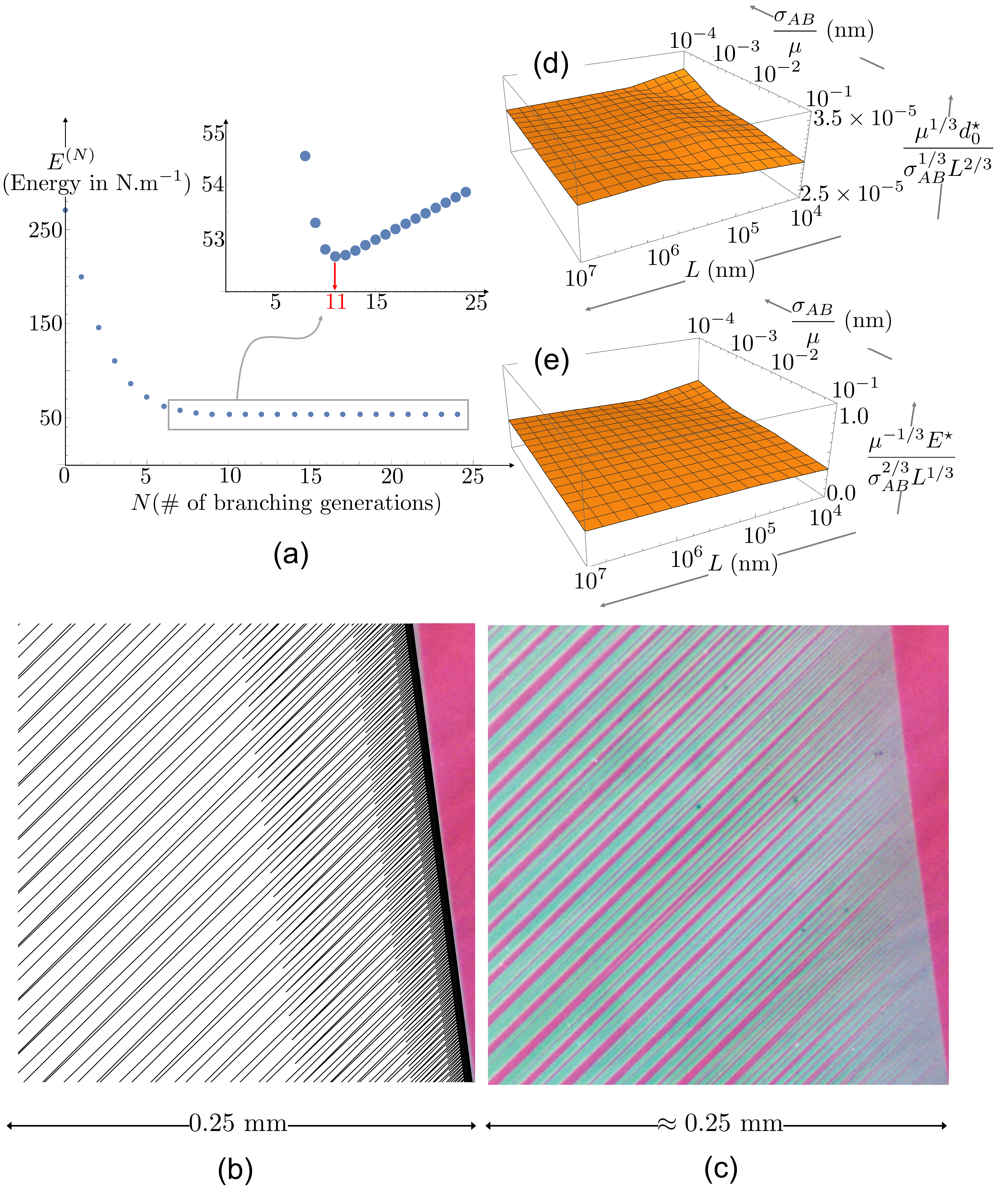}
\caption{Modeling the branched microstructure of CuAlNi (micrograph courtesy of  C.\;Chu). (a) The energy of branching microstructure for $\mu = 70$ GPa, $\sigma_{AB}  = 100$ mJ.m$^{-2}$, $L = 5$ mm, and for a given number of branching generations.  The minimum is achieved at $11$ branching generations.  This minimum yields the one-to-one model and experiment comparison in (b-c).  The parametric study in (d-e) reveals that the energy collapses to the Kohn and M\"{u}ller scaling over four orders of magnitude in geometric and material length scales.}
\label{fig:Sigma200mu70}
\end{figure}

 In order to explore the behavior of our construction, we also performed a parametric study for the `material length scale' $\sigma_{AB}/ \mu \in \{10^{-4},10^{-3},10^{-2},10^{-1} \}$ nm and the geometric length scale $L \in \{10^{4},10^{5},10^{6},10^{7}\}$ nm to address the entire range of plausible parameters. The results, i.e. the optimized values of $d^{\star}_0$, $L^{\star}_0$ (normalized with respect to the total length), $N^{\star}$ and the resulting energy $E^\star$ (normalized with respect to the elastic modulus) for each combination of parameters, are given in 
 Tab.~\ref{tab:interfaces}. 
   
The results in this table reveal several important properties of the optimized construction. Most importantly, the number of branching generations stays within reasonable limits
for the whole range of parameters. Just one generation of branching is observed for the highest $\sigma_{AB}/\mu$ ratio and the shortest geometric length, which is consistent with the assumption that the branching becomes more and more energetically preferred with decreasing $\sigma_{AB}$ and increasing length. The highest number $N^{\star}$ obtained throughout the study is 14. The twin width for the same set of parameters is $d_0^{\star}=66.1$ $\mu$m, which means that the finest laminate at the habit plane has the width of approximately $d_{14}^{\star}=4$ nm, i.e., still not below the lattice parameters for typical shape memory alloys, which is about ten times smaller. This means that, even for the highest number of branching generations, the continuum theory description may still be valid and well justified. 

\renewcommand{\arraystretch}{1.75}
\begin{table}[!ht]
	\centering

		\caption{A parametric study of energy-minimizing constructions of branching in  Cu$_{69}$Al$_{27.5}$Ni$_{3.5}$ in geometric lengths $L$ and material lengths $\sigma_{AB}/\mu$. The outputs of the study are the following optimized parameters: twin width in the unbranched region $d_0^{\star}$, the relative length of the unbranched region ${L^{\star}_0}/{L}$, and the number of branching generations $N^{\star}$. For each optimized construction, also the resulting energy $E^{\star}$ is shown.}\label{tab:interfaces}
		\begin{tabular}{p{3.5cm}p{3cm}p{3cm}p{3cm}p{2.5cm}}
			\hline \hline
		  \small{}	  &\small{$L = 10^{4}$ nm}				    &  \small{$L = 10^{5}$ nm}		          & \small{$L= 10^{6}$ nm} 		& \small{$L= 10^7$ nm} 	\\
			\hline	
			 {$\frac{\sigma_{AB}}{\mu} =10^{-4}$ nm}   & \small{$d_0^{\star} =0.7$ $\mu$m} & \small{$d_0^{\star} =3.1$ $\mu$m} 		  &\small{$d_0^{\star} = 14.3$ $\mu$m }& \small{$d_0^{\star} =66.1$ $\mu$m}           \\
			\small{}     & \small{$\frac{L^{\star}_0}{L} = 0.66$ }& \small{$\frac{L^{\star}_0}{L} = 0.66$} 		  &\small{$\frac{L^{\star}_0}{L} = 0.66$}  & \small{$\frac{L^{\star}_0}{L} = 0.66$} 	 	\\
	 \small{}		    & \small{$N^{\star}= 8$} & \small{$N^{\star}=10$} 		  &\small{$N^{\star}= 12$} & \small{$N^{\star}= 14$} 	 	\\
			    & \small{$\frac{E^{\star}}{\mu}=0.02$ nm }& \small{$\frac{E^{\star}}{\mu}=0.03$  nm} 		  &\small{$\frac{E^{\star}}{\mu}=0.07$ nm} & \small{$\frac{E^{\star}}{\mu}= 0.16$ nm}	 	\\		
			\hline		  
		{$\frac{\sigma_{AB}}{\mu} =10^{-3}$ nm}	    & \small{$d_0^{\star} =1.4$ $\mu$m} & \small{$d_0^{\star} =6.7$ $\mu$m }		  &\small{$d_0^{\star} =30.4$ $\mu$m} & \small{$d_0^{\star} = 142.5$ $\mu$m}           \\
			\small{  }     & \small{$\frac{L^{\star}_0}{L} = 0.65$}  &\small{$\frac{L^{\star}_0}{L} = 0.66$}  		  &\small{$\frac{L^{\star}_0}{L} = 0.66$} &\small{$\frac{L^{\star}_0}{L} = 0.66$}  	 	\\
			\small{ }    & \small{$N^{\star}= 6$} & \small{$N^{\star}=8$}		  &\small{$N^{\star}=10$} & \small{$N^{\star}=12$} 	 	\\	
			  		& \small{$\frac{E^{\star}}{\mu}=0.07$ nm} & \small{$\frac{E^{\star}}{\mu} =  0.16$ nm}		  &\small{$\frac{E^{\star}}{\mu}=0.35$ nm} & \small{$\frac{E^{\star}}{\mu}=0.75$ nm} 	 	\\
			\hline				
			{$\frac{\sigma_{AB}}{\mu} =10^{-2}$ nm}  & \small{$d_0^{\star} =3.0$ $\mu$m} & \small{$d_0^{\star} =14.5$ $\mu$m }		  &\small{$d_0^{\star} =66.5$ $\mu$m} & \small{$d_0^{\star} = 307.4$ $\mu$m }          \\
			\small{ }     & \small{$\frac{L^{\star}_0}{L} =0.67$} & \small{$\frac{L^{\star}_0}{L} = 0.65$} 		  &\small{$\frac{L^{\star}_0}{L} =0.66$} & \small{$\frac{L^{\star}_0}{L} = 0.66$ }	 	\\
			   \small{}  & \small{$N^{\star}=3 $} & \small{$N^{\star}=6$} 		  &\small{$N^{\star}=8$} & \small{$N^{\star}= 10$} 	 	\\	
			     	 & \small{$\frac{E^{\star}}{\mu}=0.34$ nm}& \small{$\frac{E^{\star}}{\mu}=0.74$ nm} 		  &\small{$\frac{E^{\star}}{\mu}=1.61$ nm} & \small{$\frac{E^{\star}}{\mu}= 3.47$ nm} 	 	\\
			\hline				  
			  {$\frac{\sigma_{AB}}{\mu} =10^{-1}$ nm}    & \small{$d_0^{\star} =6.3$ $\mu$m} & \small{$d_0^{\star} =29.9$ $\mu$m }		  &\small{$d_0^{\star} =144.9$ $\mu$m} & \small{$d_0^{\star} =665.3$ $\mu$m}           \\
			\small{}     & \small{$\frac{L^{\star}_0}{L} = 0.68$} & \small{$\frac{L^{\star}_0}{L} =0.67$} 		  &\small{$\frac{L^{\star}_0}{L} =0.65$} & \small{$\frac{L^{\star}_0}{L} = 0.66$} 	 	\\
			  \small{}   & \small{$N^{\star}=1 $} & \small{$N^{\star}=3$} 		  &\small{$N^{\star}= 6$} & \small{$N^{\star}=8$} 	 	\\	
			 & \small{$\frac{E^{\star}}{\mu}=1.44$ nm} & \small{$\frac{E^{\star}}{\mu}=3.39$ nm}		  &\small{$\frac{E^{\star}}{\mu}=7.44$ nm}& \small{$\frac{E^{\star}}{\mu}=16.16$ nm }	\\	\hline \hline
		\end{tabular}
		
\end{table}

The decisive factor for $N^{\star}$ appears to be ratio between the material and geometric lengths, $\sigma_{AB}/\mu{}L$. For example, all combinations of parameters with $\sigma_{AB}/\mu{}L=10^{-8}$ give $N^{\star}=8$. This is a natural result, as the ratio $\sigma_{AB}/\mu{}L$ measures, in some sense, how effectively the total energy is reduced with each branching generation. 
A more surprising result is obtained for the length of the unbranched segment $L_0$. We observe that, for the whole range of parameters explored, this length is approximately equal to two thirds of the total length. This is consistent with the fact that branching is typically experimentally observed just close to the habit planes, while many laminates in temperature-induced martensitic microstructures are unbranched.

In addition, we find that the data for the optimal unbranched twin width $d_0^{\star}$ and energy $E^{\star}/\mu$ nearly collapse to constants when these quantities are normalized by $\mu^{-1/3} \sigma_{AB}^{1/3}L^{2/3}$ and $\mu^{-2/3} \sigma_{AB}^{2/3} L^{1/3}$, respectively (Fig.~\ref{fig:Sigma200mu70}(d-e)).  This observation suggests that the scaling laws for our construction are
\begin{align}\label{scalings}
d_0^{\star} \sim \mu^{-1/3} \sigma_{AB}^{1/3} L^{2/3}, \quad E^{\star} \sim \mu^{1/3} \sigma_{AB}^{2/3} L^{1/3}.
\end{align}
These nontrivial scalings agree with the results of  Kohn and M\"{u}ller in the simplified (non-vectoral) setting for $\sigma_{AB}\rightarrow{0}$, but Kohn and M\"{u}ller's derivation (and other derivations in this direction) made somewhat coarser approximations for upper and lower bounds.  
We  show in the Appendix that the energy of the construction used for the above numerical test is bounded above by $C\mu^{1/3}\sigma^{2/3}L^{1/3}$, where $C$ is a constant dependent only on the crystallographic parameters (which means independent of $\mu$, $\sigma_{AB}$ and $L$), and that this bound holds not only in the $\sigma_{AB}\rightarrow{0}$ limit, but also for any $\sigma_{AB}<\mu{}L$, i.e., whenever the above introduced material length is smaller than the physical length of the martensite region. The fact that the rescaling of the results of the numerical test appears to collapse to nearly a constant, i.e., $d_0^{\star} \approx 3 \times 10^{-5} \mu^{-1/3} \sigma_{AB}^{1/3} L^{2/3}$ and $E^{\star} \approx 0.35  \mu^{1/3} \sigma_{AB}^{2/3} L^{1/3}$ for the entire range of experimentally/physically relevant geometric and material length scales suggests that the energy of our construction may be close to the $C\mu^{1/3}\sigma^{2/3}L^{1/3}$ upper bound.  

In the Appendix, we also present an ansatz-free lower bound for the full three-dimensional setting, constructed in the same spirit as  \cite{KM1,KM2} and  \cite{ChanConti1,ChanConti2}. The result is that the energy of the branched structure for $\sigma_{AB}<\mu{}L$ is bounded below by $c\mu^{1/3}\sigma^{2/3}L^{1/3}$, where $c$ is another constant, $0<c\leq{}C$.  This directly implies that, for $\sigma_{AB}\rightarrow{0}$, our construction gives the optimal energy scaling.  Nevertheless, neither the upper bound nor the lower bound  imply that the energy of a real branched microstructure with finite $N$ and small but finite $\sigma_{AB}$ should follow the $\mu^{1/3}\sigma^{2/3}L^{1/3}$ scaling. As seen in Tab.~\ref{tab:interfaces}, only one generation of branching is optimal (within our ansatz)  for $\sigma_{AB}/\mu{}L=10^{-5}$, i.e., in a setting where $\sigma_{AB}\ll\mu{}L$. For higher $\sigma_{AB}$, one can expect the simple laminate microstructure with the closure domains to be the optimal upper bound within our ansatz, and this simple laminate can be expected to exhibit the $E\sim\sigma_{AB}^{1/2}$ scaling with $\sigma_{AB}\rightarrow\mu{}L$. However, the lower bound still holds, which indicates that $c$ and $C$ must be relatively far away from each other, and so the conclusion on the scaling can be drawn indeed only for  $\sigma_{AB}\rightarrow{0}$. On the other hand, the scaling argument as formulated in Section \ref{scalingargument}  indeed holds: regardless of the prefactors, the branching construction is always energetically preferred over a simple 1$-$st order laminate in the $\sigma_{AB}\rightarrow{0}$ limit. This is, however, not  surprising, as  $\sigma_{AB}\rightarrow{0}$ implies $N^{\star}\rightarrow{\infty}$, and the detailed construction becomes nearly identical to the simple construction presented in Section 3, where $N$ was considered sufficiently large and the energy of the closure domains was neglected.

\section{Conclusions}

The main aim of this work was to construct a realistic model for the energy and length scales of a branched martensitic microstructure, directly applicable to real shape memory alloys, and to study the properties of this model, mainly in terms of its energy scaling. This required that the construction be done in a fully three-dimensional setting of non-linear elasticity. Motivated by the classical approach of  \cite{KM1,KM2}, we used a self-similar construction for the microstructure, and showed that the resulting energy upper bound gives the expected scaling $E\sim{}\mu^{1/3} \sigma_{AB}^{2/3} L^{1/3}$; thus, the fundamental scaling argument for small interfacial energy ($\sigma_{AB}\rightarrow{0}$) holds. Hence, the energy of the branched structures is, at least in the limit, lower than the energy of a simple laminate -- this result was expected from previous simplified constructions \citep{KM1,KM2,Capella_Otto_1,Capella_Otto_2,Dondl,Conti}, but here we proved its validity for a three dimensional construction applicable to  real alloys. 

The proposed construction is versatile, enabling several additional features to be incorporated into the model. One example is the more detailed construction introduced in Section 4 and developed in the Appendix; this construction takes into account the fact that the 0$-$th layer remains unbranched, does not neglect the energy of the closure domains, and, most importantly, it anticipates that the number of branching generations follows from energy minimization. Although this construction is still rather a rough upper bound, i.e., stress equilibrium is not imposed and the elastic strains are piecewise homogeneous, we show that the geometric parameters  -- the width of the twins, the number of the branching generations, the aspect ratios of the branching segments --  corresponding to the energy minimizer within the ansatz of the construction mimic very well experimental observations. This suggests that the delicate balance between elastic and surface energies is the dominant mechanism for the observed microstructures at the austenite-martensite interface in shape memory alloys.  Our conclusions give insight on how these structures can be manipulated.

The subsequent parametric study revealed that the detailed construction obeys relatively simple scaling laws for the energy and the twin widths throughout a broad range of physically admissible parameters, which covers $N^{\star}$ ranging from 1 to 14. This is a rather unexpected result, as the construction involves balancing of the energy between the branched microstructure and the closure domains.

\section*{Acknowledgment}

H.S. and B.B. acknowledge  financial support from the Czech Science Foundation [grant no. 18-03834S]. H.S. further thanks  J. William Fulbright Commission (Prague) and the Ministry of Education, Youth and Sports of the Czech Republic, grant program INTER-EXCELLENCE/INTER-ACTION [grant No. LTAUSA18199].  PP thanks the MURI program (FA9550- 16-1-0566) for support.
RDJ benefited from the support of NSF (DMREF-1629026), ONR (N00014-18-1-2766), MURI (FA9550-18-1-0095), the Medtronic Corp. and a Vannevar Bush Faculty Fellowship.

\appendix

\section{Construction and full optimization of a detailed model}

In this appendix, we present an example of a detailed construction of a branching microstructure, in which the deformation gradient is almost everywhere piecewise constant and rank-one compatible across interfaces. This construction includes the unbranched $0-$th layer and the closure domains, and will be fully optimized within a set of modeling assumptions (discussed below) that enables explicit forms for the optimal number of branching generations $N^{\star}$, the length of the unbranched segment $L^{\star}_0$ and subsequent branching segments $L^{\star}_1, \ldots, L^{\star}_{N^{\star}}$, and the optimal twin width $d^{\star}_0$.  These are obtained in terms of experimentally measurable quantities only: a characteristic modulus of martensite $\mu$, the crystallographic parameters $\mathbf{a}, \mathbf{b}, \mathbf{m}, \mathbf{n},\alpha = \sqrt{1- (\mathbf{n} \cdot \mathbf{m})^2}$ and $\lambda$ (see (\ref{compat1}-\ref{compat2})), the interfacial energy of a twin per unit length $\sigma_{AB} = \sigma |\mathbf{a}|$ (see equation (\ref{assump1}) below), and the length of the overall $AB$ laminate $L$.  We identify these as the \textit{experimental parameters} throughout the calculation below.  

\subsection{Formulation}

The detailed model reflects the growth of the surface energy with increasing $\varepsilon$. This increase appears, physically, due to two reasons: i) the interfaces are becoming longer, and ii) the interfaces are rotated from the crystallographically preferred twinning planes, which can be understood as formation of defects (steps) on the twinning plane that costs additional energy. Principally, this defect formation need not be a consequence of elastic strain.  However, it happens that the surface energy term modeled by relation (\ref{surfE})
 \begin{equation}\label{assump1}
E_{\text{surf.}}[\mathbf{y}] =  \int_{\Omega} \sigma |\nabla^2 \mathbf{y}| dx = \sigma \times  (``\text{length of the interface}")  \times (``\text{jump in the def.\;gradient}")
 \end{equation}
 does, in fact, reasonably account for these coupled effects since the jump in the deformation gradient does increase with increasing $\varepsilon$. Note that the interfacial energy per unit length of a twin obtained from experiments is $\sigma_{AB} = \sigma |\mathbf{a}|$  in this setting. Note further that, for constant deformation gradients ${\mathbf A}$ and ${\mathbf B}$ forming a compatible planar interface, the jump in the deformation gradient appearing in equation (\ref{assump1}) is equal to $|{\mathbf A}-{\mathbf B}|$, where $|\cdot|$ is the Frobenius norm. Since two such gradients must satisfy equation (\ref{compat1}) with some shearing vector ${\mathbf a}$, the Frobenius norm of ${\mathbf A}-{\mathbf B}$ is directly equal to $|{\mathbf a}|$.
 
We also use a simplified, but physically justified, model for the elastic energy density $\varphi(\mathbf{F})$ of the shape memory alloy to make the calculation fully explicit.  In this direction, it is natural to consider the energy as a linear elastic perturbation away from the deformation gradients $\mathbf{A}$ and $\mathbf{B}$ since branching microstructure should not deviate significantly from these deformation gradients.  The natural perturbation is akin to $\varphi_E$ in (\ref{varphiE}), but, given the lack of a precise characterization of the modulus tensor of martensite for most shape memory alloys, we instead introduce the single moduli approximation
\begin{align}\label{assump2}
\varphi(\mathbf{F})  \stackrel{\rm def.}{=} \frac{\mu}{2} \min\Big\{  |\sym \boldsymbol{\delta} \mathbf{A} |^2, |\sym \boldsymbol{\delta} \mathbf{B}|^2\Big\} \quad \text{ where } \quad \begin{cases}  \boldsymbol{\delta}\mathbf{A}  \stackrel{\rm def.}{=} \mathbf{F} - \mathbf{A} \\
\boldsymbol{\delta} \mathbf{B}  \stackrel{\rm def.}{=} \mathbf{F} - \mathbf{B}   \end{cases}
\end{align}
to make an explicit calculation as simple as possible. Here, we think of $\mu$ as a characteristic elastic modulus of the martensite phase.

\subsection{Derivation of the total energy}\label{ssec:derivation}

\begin{figure}[t!]
\centering
\includegraphics[width= 6.8 in]{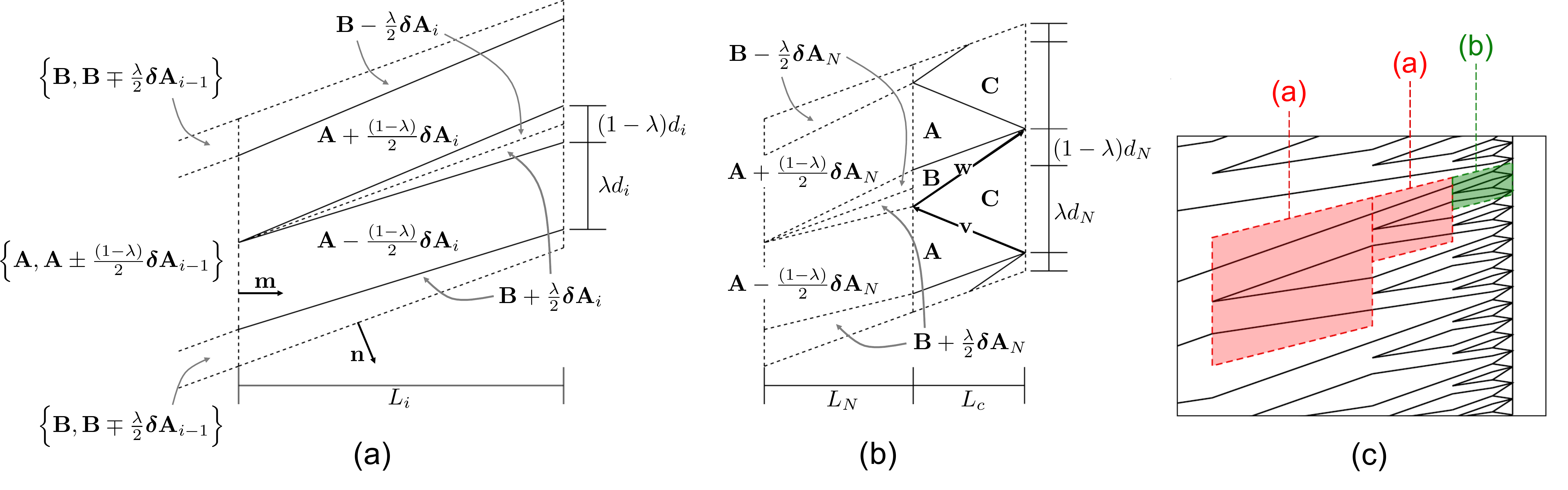}
\caption{Ansatz for the construction of branching microstructure in the detailed model.  (a).  Each branching layer is constructed with all the interfaces slightly tilted using the perturbations $\boldsymbol{\delta} \mathbf{A}_i = \varepsilon_i \mathbf{a} \otimes \mathbf{m}$.  The magnitude and sign of this perturbation in each layer is chosen to give the compatible piecewise homogenous deformations gradients with the smallest energy density $E_{\rm elast.}^{(i)}$;  (b). The geometry of a closure domain of arbitrary length $L_c > 0$.  This has a deformation gradient $\mathbf{C}$ (explicitly given below) that is simultaneously compatible with $\mathbf{A}$ along $\mathbf{v}$, $\mathbf{B}$ along $\mathbf{w}$, and $\mathbf{I}$ along $\mathbf{m}^{\perp}$ (in the plane spanned by $\mathbf{m}$ and $\mathbf{n}$); (c) placements of (a) and (b) segments in the construction.}
\label{BestAnsatz}
\end{figure}

The total energy (\ref{energy}) for the construction, computed using (\ref{assump1}) and (\ref{assump2}), has five distinct contributions 
\begin{equation}\label{0Result}
E = E_{0} + E_{\rm elast.}^{br.} + E_{\rm surf.}^{br.} + E_{\rm elast.}^{cl.} + E_{\rm surf.}^{cl.}
\end{equation}
obtained as follows:
\begin{itemize}
 \item {\it The unbranched $0$-th layer $-$} As discussed above, there is no benefit to branching in this layer; so we may use a simple $AB$ laminate here.   For the width of a twinning domain $d_0>0$ and length of the layer $0 \leq L_0 < L$, the surface energy of the $0$-th unit cell is then
 \begin{equation}
E_{\rm surf.}^{(0)} = 2 \sigma \frac{L_0}{\alpha} |\mathbf{A} - \mathbf{B}| = 2 \sigma |\mathbf{a}|\frac{L_0}{\alpha} = 2 \sigma_{AB}\frac{L_0}{\alpha}
\end{equation}
due to the rank-one compatibility condition (\ref{compat1}). There is no elastic energy in the $AB$ laminate.  Consequently, the total energy of the layer is
\begin{equation}\label{1Result}
E_0 = 2 \sigma_{AB}  \frac{L_0}{\alpha d_0}  = 2  (1-\lambda) \sigma_{AB} \varepsilon_0^{-1} 
\end{equation}
since there are $1/d_0$ unit cells in the layer.  (Here and below, we recall the definition (\ref{diLi}).) 
\item {\it The elastic energy of branching $-$}  To accommodate branching, we introduce a compatible, but strained, set of deformation gradients shown schematically in Fig.~\ref{BestAnsatz}(a).  These are $\mathbf{A} \pm \frac{(1-\lambda)}{2}\boldsymbol{\delta}\mathbf{A}_i$ each in an area $\lambda d_i L_i$ and $\mathbf{B} \pm \frac{\lambda}{2} \boldsymbol{\delta} \mathbf{A}_i$ each in an area $(1-\lambda) d_iL_i$ on the $2 d_i L_i$ cell.  As $\boldsymbol{\delta} \mathbf{A}_i = \varepsilon_i \mathbf{a} \otimes \mathbf{m}$, the total elastic energy of the cell is thus
\begin{equation}
\begin{aligned}
&E_{\rm elast.}^{(i)} = \frac{\mu}{4} \lambda(1-\lambda)L_i d_i |\sym(\boldsymbol{\delta} \mathbf{A}_i)|^2=  \frac{\mu}{4}\lambda(1-\lambda) g_2 L_i d_i \varepsilon_i^2 
\end{aligned}
\end{equation}
for  $g_2  \stackrel{\rm def.}{=} \frac{1}{2}(|\mathbf{a}|^2 + (\mathbf{a} \cdot \mathbf{m})^2)$.  To obtain the first equality, we argue on physical grounds that the tilt $\varepsilon_i$ will always remain sufficiently small so that the minimizing perturbations for the energy density (\ref{assump2}) are as formulated in the figure. Indeed, large tilt away from the crystallographic plane is never observed experimentally, and so we expect this to hold true also for energy minimizing configurations of the detailed model described herein. This leads directly to the \textit{geometric factor} $g_2$ as shown.  Hence, as there are $1/d_{i-1} = 1/(2d_i) =2^{i}/(2 d_0)$ cells in each layer, the elastic energy of the branched part of the domain is given by 
\begin{equation}
\begin{aligned}\label{2Result}
E_{\rm elast.}^{br.} &= \frac{\mu}{8} \lambda(1-\lambda)  g_2 \sum_{i = 1}^N L_i \varepsilon_i^2 = \frac{\mu}{8} \lambda(1-\lambda)^2 \alpha g_2 \Big(d_0 \sum_{i=1}^N 2^{-i} \varepsilon_i\Big).
\end{aligned}
\end{equation}
Notice that every layer of the proposed microstructure  in Fig.~\ref{BestAnsatz} satisfies the macroscopic compatibly condition with austenite, as  
 \begin{equation}\label{eq:moved}
 \begin{aligned}
 &\frac{\lambda}{2} \Big( \mathbf{A} - \frac{(1-\lambda)}{2} \boldsymbol{\delta}\mathbf{A}_i\Big) + \frac{\lambda}{2}  \Big( \mathbf{A} - \frac{(1-\lambda)}{2} \boldsymbol{\delta}\mathbf{A}_i\Big) + \frac{(1-\lambda)}{2} \Big( \mathbf{B} - \frac{\lambda}{2} \boldsymbol{\delta}\mathbf{A}_i\Big) + \frac{(1-\lambda)}{2}\Big( \mathbf{B} + \frac{\lambda}{2} \boldsymbol{\delta}\mathbf{A}_i\Big) \\
 &\qquad = \lambda \mathbf{A} + (1-\lambda) \mathbf{B} = \mathbf{I} + \mathbf{b} \otimes \mathbf{m},
 \end{aligned} 
 \end{equation}
 with the $0-$th unbranched layer having  $\boldsymbol{\delta} \mathbf{A}_0 = \mathbf{0}$.
\item {\it The surface energy of branching $-$} This calculation is more involved as there are many interfaces and subcases to consider for the cell problem in Fig.~\ref{BestAnsatz}(a): (i)  two interfaces of length $|\alpha^{-1} L_i \mathbf{n}^{\perp} + 2^{-1}(1-\lambda) d_{i} \mathbf{m}^{\perp}|$ with jump across the interface $|2^{-1} \boldsymbol{\delta} \mathbf{A}_i + \mathbf{A} - \mathbf{B}|$;  (ii) two of length $|\alpha^{-1} L_i \mathbf{n}^{\perp} - 2^{-1}(1-\lambda) d_{i} \mathbf{m}^{\perp}|$ with jump $|-2^{-1}\boldsymbol{\delta} \mathbf{A}_i + \mathbf{A} - \mathbf{B}|$;  (iii) two of length $\alpha^{-1}L_i$ with jump $|\lambda \boldsymbol{\delta} \mathbf{A}_i|$; (iv) and one of length $2d_i$ describing the transition from layer $i-1$ to layer $i$.  This latter interface has three potential subcases (as shown in Fig. \ref{BestAnsatz}(a)), but a straightforward calculation reveals that, if $\varepsilon_{i-1} \leq \varepsilon_i$, then each subcase has the exact same surface energy: $\sigma |\mathbf{a}|(1-\lambda) \lambda d_i \varepsilon_i$.   This monotonicity assumption is reasonable, as it is suggested from the optimization in the self-similar analysis (\ref{epsopt}). 
  In enforcing this assumption and combining (i-iv), the total surface energy of the cell problem is given by 
\begin{align}
E_{\rm surf.}^{(i)} &=  2 \sigma |\mathbf{a}| \alpha^{-1} L_i\Big( |\mathbf{n}+ \frac{\varepsilon_i}{2} \mathbf{m}|^2 +  |\mathbf{n} - \frac{\varepsilon_i}{2} \mathbf{m}|^2\Big)  + 2  \lambda \sigma |\mathbf{a}| \alpha^{-1} L_i \varepsilon_i +  \sigma |\mathbf{a}|(1-\lambda) \lambda d_i \varepsilon_i  \nonumber \\
&= 2 \sigma_{AB} (1-\lambda ) d_i \Big(\frac{2}{\varepsilon_i} + \lambda + \frac{1+\lambda}{2} \varepsilon_i  \Big).
\end{align}
The first term in the first equality is the result of a explicit calculation involving (i) and (ii) above, the second follows an explicit calculation involving (iii), and the third is (iv).   The second equality is then obtained directly from the first by substituting $L_i = (1-\lambda) \alpha d_i/\varepsilon_i$. Consequently and as there are $1/d_{i-1} = 1/(2d_i) =2^{i}/(2 d_0)$ cells in each layer, the total surface energy of the branched part of the domain is given by 
\begin{equation}\label{3Result}
E_{\rm surf.}^{br.} =  (1-\lambda )\sigma_{AB} \sum_{i = 1}^{N} \Big(\frac{2}{\varepsilon_i} + \lambda + \frac{1+\lambda}{2} \varepsilon_i  \Big).
\end{equation}
 \item {\it The elastic energy of closure domains $-$} After $N$ branching generations, we utilize closure domains to make the branching structure compatible with the austenite.  Following \cite{zhang}, we introduce a triangular region of length $L_c >0$ and deformation gradient $\mathbf{C}$ that is compatible with the $AB$ laminate and the austenite (see Fig.~\ref{BestAnsatz}(b)).  This satisfies
 \begin{equation}
\begin{aligned}
\mathbf{C} &= \mathbf{I} + \mathbf{b} \otimes \mathbf{m} + \lambda \varepsilon_c \mathbf{a} \otimes \mathbf{m} 
\end{aligned}
\end{equation}
for $\varepsilon_c   \stackrel{\rm def.}{=}  (1-\lambda)\alpha \frac{d_N}{L_c}$,
and it is compatible along the edges 
\begin{equation}
\begin{aligned}
&\mathbf{v} = \lambda d_N \mathbf{m}^{\perp} - \alpha^{-1} L_c\mathbf{n}^{\perp}  \quad (\text{with } \mathbf{A}),  \\
&\mathbf{w} = (1-\lambda) d_N \mathbf{m}^{\perp} + \alpha^{-1} L_c \mathbf{n}^{\perp} \quad (\text{with } \mathbf{B}).
\end{aligned} 
\end{equation}
As a result, exactly one half of the area of cell problem (i.e., $\frac{1}{2} L_c d_N$) has an elastic energy with energy density $\varphi(\mathbf{C})$.  The other half has zero elastic energy.  This results in an elastic energy for the cell of
\begin{equation}
\begin{aligned}
E^{(N+1)}_{elast.} &=  \frac{\mu}{4} L_c d_N \min_{\xi \in \{\lambda, -(1-\lambda)\}} |\sym\big(\mathbf{a}\otimes(\xi \mathbf{n} + \lambda \varepsilon_c \mathbf{m})\big)|^2\\
&=  \frac{\mu}{4} L_c d_N\Big( g_0 \check{\xi}^2 + 2 g_1\check{\xi} \lambda  \varepsilon_c +   g_2\lambda^2 \varepsilon_c^2 \Big), 
\end{aligned}
\end{equation}
for $g_0 \stackrel{\rm def.}{=}  \frac{1}{2} (|\mathbf{a}|^2 + (\mathbf{a} \cdot \mathbf{n})^2)$, $g_1\stackrel{\rm def.}{=}  \frac{1}{2} \big(|\mathbf{a}|^2(\mathbf{n} \cdot \mathbf{m}) + (\mathbf{a} \cdot \mathbf{m})(\mathbf{a} \cdot \mathbf{n})\big)$ and $g_2$ defined previously.  Here, $\check{\xi} \in  \{\lambda, -(1-\lambda)\}$ denotes the minimizer above. Since there are $1/d_N = 2^N/d_0$ cells in the closure domain layer, we conclude that 
\begin{align}\label{4Result}
E^{cl.}_{elast.} &= \frac{\mu}{4} (1-\lambda) \alpha \frac{d_0}{2^N} \Big( \frac{g_0\check{\xi}^2}{\varepsilon_c} + 2 g_1 \check{\xi} \lambda + g_2 \lambda^2 \varepsilon_c \Big).
\end{align}
\item  {\it The surface energy of closure domains $-$} This has at most\footnote{Technically, the last two contributions vanish in the case that we are dealing with a simple $AB$ laminate (i.e., the case $N=0$). We ignore this very slight difference in the formulas (\ref{surfClose}) and (\ref{5Result}), as it will not influence the results.} five contributions in the cell problem (Fig.~\ref{BestAnsatz}(b)): one of length $\alpha^{-1} L_c$ with jump $|\mathbf{A} - \mathbf{B}|$; one of length $|\mathbf{v}|$ with jump $|\mathbf{C} - \mathbf{A}|$; and one of length $|\mathbf{w}|$ of jump $|\mathbf{C} -\mathbf{B}|$; one of length $\lambda d_N$ with jump $\frac{(1-\lambda)}{2}|\boldsymbol{\delta} \mathbf{A}_N|$; and one of length $(1-\lambda) d_N$ with jump $\frac{\lambda}{2}|\boldsymbol{\delta} \mathbf{A}_N|$.   Consequently,  the cell problem satisfies
\begin{equation}
\begin{aligned}\label{surfClose}
E^{(N+1)}_{surf.} = \sigma_{AB} L_c \alpha^{-1} \Big(2 + \frac{\lambda}{(1-\lambda)} \varepsilon_c^2 \Big) + \sigma_{AB} \lambda(1-\lambda) d_N\varepsilon_N
\end{aligned}
\end{equation}
after some routine algebraic manipulations similar to the preceding calculations.  As there are $1/d_N = 2^N/d_0$ cell problems in the closure domain layer, we conclude
\begin{equation}
\begin{aligned}\label{5Result}
E^{cl.}_{surf.} &=  \sigma_{AB} \Big( 2 (1-\lambda)\frac{1}{\varepsilon_c}  + \lambda \varepsilon_{c}  +\lambda(1-\lambda) \varepsilon_N \Big) .
\end{aligned}
\end{equation}
\item  {\it The total length is $L$ $-$} Notice that the energies above are not described in terms of the lengths $L_0,\ldots,L_N, L_c$, but rather in terms of the aspect ratios $\varepsilon_0, \ldots, \varepsilon_N, \varepsilon_c$.  In light of this, we treat the aspect ratios $\varepsilon_i$ and $\varepsilon_c$ as the independent variables, and treat the lengths as dependent variables that obey the formulas $L_i = (1-\lambda) \alpha d_i/\varepsilon_i$ and $L_c = (1-\lambda) \alpha d_N/\varepsilon_c$.  There is an additional constraint:  The total length of the twinned martensitic  region including the closure domain is fixed to be $L$, i.e.,
\begin{equation}
L = \sum_{i = 0}^N L_i + L_c.
\end{equation}
This can equivalently be written in the form of a constraint on $\varepsilon_0$, i.e.,
\begin{equation}\label{6Result}
\varepsilon_0 \equiv \varepsilon_0(\varepsilon_1,\ldots, \varepsilon_N, \varepsilon_c, d_0) = \Big( \frac{L}{(1-\lambda) \alpha d_0} - \sum_{i=0}^N \frac{1}{\varepsilon_i 2^i} - \frac{1}{\varepsilon_c 2^N} \Big)^{-1},
\end{equation}
using the formulas for the lengths.  Here, $\varepsilon_0$ can take the value $\pm\infty$, which is equivalent to $L_0 = 0$ (i.e., no unbranched layer). On the other hand, if the fixed length $L$ is a finite number, $\varepsilon_0$ can be never equal to zero, which would mean $L_0$ going to infinity. 
\end{itemize}
In collecting the calculations (\ref{1Result}), (\ref{2Result}), (\ref{3Result}), (\ref{4Result}), (\ref{5Result}), (\ref{6Result}), the total energy of branching microstructure (\ref{0Result}) is thus 
\begin{equation}\label{energyParam}
E \equiv E(\varepsilon_1,\ldots, \varepsilon_N, \varepsilon_c, d_0) =  \tilde{E}(\varepsilon_0(\varepsilon_1,\ldots, \varepsilon_N, \varepsilon_c, d_0), \varepsilon_1, \ldots, \varepsilon_N, \varepsilon_c, d_0).
\end{equation}
These parameters can all be freely optimized. 

\subsection{Optimization of the parameters}\label{ssec:optimization}

We make the following observation regarding the aspect ratios for any $d_0 >0$:
\begin{itemize} 
\item {\it For the closure domains $-$} The conditions $\frac{\partial E}{ \partial \varepsilon_c} = \frac{\partial \tilde{E}}{\partial \varepsilon_0} \frac{\partial \varepsilon_0}{\partial \varepsilon_c}  + \frac{\partial \tilde{E}}{\partial \varepsilon_c}= 0$ and  $\varepsilon_c > 0$ are equivalent to 
\begin{align}\label{epcd0}
\varepsilon_c =  \sqrt{\frac{2\sigma_{AB} (2^N-1) + \frac{\mu}{4} \check{\xi}^2 \alpha g_0 d_0 }{\sigma_{AB} \frac{\lambda}{1-\lambda} 2^N + \frac{\mu}{4} \lambda^2 \alpha g_2 d_0}}  \stackrel{\rm def.}{=}  \varepsilon_c(d_0).
\end{align}
\item {\it For most of the branching layers $-$} If $i \in \{ 1,\ldots, N-1\}$, the conditions $\frac{\partial E}{ \partial \varepsilon_i} = \frac{\partial \tilde{E}}{\partial \varepsilon_0}  \frac{\partial \varepsilon_0}{\partial \varepsilon_i} + \frac{\partial \tilde{E}}{\partial \varepsilon_i}= 0$ and $\varepsilon_i > 0$ are equivalent to 
\begin{equation}\label{epid0}
\varepsilon_i = \sqrt{\frac{16 \sigma_{AB} (2^i-1)}{\mu \lambda (1-\lambda) \alpha g_2 d_0 + (1+\lambda) \sigma_{AB}2^{i+2}} } \stackrel{\rm def.}{=}  \varepsilon_i(d_0).
\end{equation}
Here, $0< \varepsilon_1(d_0) < \varepsilon_2(d_0) < \ldots < \varepsilon_{N-1}(d_0) <2 $ for all  $d_0> 0$.
\item {\it For the final branching layer $-$} The conditions  $\frac{\partial E}{ \partial \varepsilon_N} = \frac{\partial \tilde{E}}{\partial \varepsilon_0}  \frac{\partial \varepsilon_0}{\partial \varepsilon_N} + \frac{\partial \tilde{E}}{\partial \varepsilon_N}= 0$ and $\varepsilon_N > 0$ are equivalent to 
\begin{equation}\label{epNd0}
 \varepsilon_N = \sqrt{\frac{16 \sigma_{AB} (2^N + \frac{\lambda}{2} -1)}{\mu \lambda (1-\lambda) \alpha g_2 d_0 + (1+\lambda) \sigma_{AB}2^{N+2}}}  \stackrel{\rm def.}{=}  \varepsilon_N(d_0).
\end{equation}
Here, $0< \varepsilon_{N-1}(d_0) < \varepsilon_{N}(d_0) <  2$ for $\varepsilon_{N-1}(\cdot)$ in (\ref{epid0}) and for all $d_0 >0$.  
\end{itemize}
These results follow quite directly from expanding out the partial derivatives for $\tilde{E}$ using the explicit definitions of the various energies involved.  As a consequence, the energy (\ref{energyParam}) achieves a critical point on the domain $\varepsilon_1,\ldots, \varepsilon_N, \varepsilon_c,d_0 >0$ if and only if 
\begin{equation}\label{critical}
(\varepsilon_1, \ldots, \varepsilon_N, \varepsilon_c) = (\varepsilon_1(d_0), \ldots, \varepsilon_N(d_0), \varepsilon_c(d_0)), \quad \frac{\partial E}{\partial d_0} =  \frac{\partial \tilde{E}}{\partial \varepsilon_0}  \frac{\partial \varepsilon_0}{\partial d_0} + \frac{\partial \tilde{E}}{\partial d_0}= 0
\end{equation}
for some $d_0 > 0$.  By an explicit calculation, the latter is then solved if and only if 
\begin{align}\label{fd0}
f_N(d_0) &  \stackrel{\rm def.}{=}  -\frac{2 \sigma_{AB} L}{\alpha d_0^2}  + \frac{\mu}{4}  \frac{(1-\lambda) \alpha}{2^N} \Big( \frac{\check{\xi}^2g_0}{\varepsilon_c(d_0)} +2\check{\xi} \lambda g_1 + \lambda^2g_2 \varepsilon_c(d_0) \Big) + \frac{\mu}{8} \lambda(1-\lambda)^2\alpha g_2  \sum_{i=1}^N 2^{-i} \varepsilon_i(d_0)    
\end{align}
vanishes for some $d_0 >0$.  This function is both continuous and strictly increasing (i.e., $f_N' >0$) on the positive real numbers. It also has the limiting behavior $\lim_{d_0 \searrow 0} f_N(d_0) = -\infty$  and\footnote{The strict inequality in  (\ref{limInfinity}) is due to the fact that $|\mathbf{a}| \neq 0$ and $\alpha \neq 0$, i.e., that fact that $\mathbf{n}$ and $\mathbf{m}$ are not parallel.  Indeed, these facts make it quite easy to see that $g_0 g_2 > g_1^2$.}
\begin{equation}\label{limInfinity}
\lim_{d_0 \nearrow \infty} f_N(d_0) = \frac{\mu}{2^{N+1}} (1-\lambda)\lambda \alpha (|\check{\xi}| \sqrt{g_0 g_2} + \check{\xi} g_1) > 0.
\end{equation}
As such, there is a unique twin width satisfying
\begin{equation}\label{d0N}
d_0^{(N)} > 0, \quad f_N(d^{(N)}_0) = 0.
\end{equation}
It is easily obtained numerically given the experimental parameters and $N$.  Finally, we need to check that the assumed monotonicity of the $\varepsilon_i$ remains preserved for the optimized quantitites, as we did not enforce it during the optimization.  {The validity of this assumption can be seen by replacing $2^{i}$ in (\ref{epid0}) with a parameter $t$, differentiating with respect to this parameter, and observing that the derivative is positive.  The monotonicity for $\varepsilon_N$ compared to $\varepsilon_{N-1}$ follows from the monotonicity of the $\varepsilon_i$.  Finally, the upper bounds are deduced by noticing that $\varepsilon_i(d_0) \leq \varepsilon_i(0)$ for all $d_0 > 0$ and $i = 1,\ldots, N$ and using the fact that $\lambda \in (0,1)$.}

\bigskip
The results (\ref{critical}) and (\ref{d0N}) furnish that  $(\varepsilon_1(d_0^{(N)}), \ldots, \varepsilon_N(d_0^{(N)}), \varepsilon_c(d_0^{(N)}), d_0^{(N)})$ is the unique global minimizer to  (\ref{energyParam}) on the domain $\varepsilon_1,\ldots, \varepsilon_N, \varepsilon_c,d_0 >0$.\footnote{Indeed, it is the unique critical point on this domain.  Further, a laborious but straightforward calculation reveals that the total energy blows up in the limit that any of the parameters approaches $0$ or $\infty$.  That is, for any $t \in \{ d_0, \varepsilon_1, \ldots, \varepsilon_N, \varepsilon_c\}$, we find that $\lim_{t \searrow 0} E = \infty$ and $\lim_{t \nearrow \infty} E = \infty.$  This guarantees that the critical point is energy minimizing.}  We write
\begin{align}\label{EN}
E^{(N)}   &\stackrel{\rm def.}{=}  E( \varepsilon_1(d_0^{(N)}), \ldots, \varepsilon_N(d_0^{(N)}), \varepsilon_c(d_0^{(N)}), d_0^{(N)}) = \inf_{\varepsilon_1, \ldots, \varepsilon_N, \varepsilon_c, d_0 >0} E(\varepsilon_1, \ldots, \varepsilon_N,\varepsilon_c, d_0).
\end{align}  
The energy of these configurations satisfies $E^{(N)} \rightarrow \infty$ as $N \rightarrow \infty$ due to the growth in surface energy with increasing branching generations.  Consequently, the configuration of minimal energy occurs at a finite number of branching generations $N^{\star} \in \{0,1,2, \ldots \}$ such that   
\begin{equation}\label{EStar}
E^{(N^{\star})} = \inf_{N \in \{0,1,2, \ldots \} }  E^{(N)} \stackrel{\rm def.}{=} E^{\star}.
\end{equation} 
By studying the regime $N \gg 1$, where there are very fine number of branching generations, we obtain a cutoff which guarantees that this optimum $N^{\star}$ is not \textit{too large}, in the sense that it can be bounded from above by a quantity that depends only on the experimental parameters.   This observation gives us a formal guarantee for a numerical procedure to evaluate the optimum $N^{\star}$.

To develop this result, notice that the equality (\ref{d0N}) furnishes the inequality
\begin{equation}
\frac{32 \sigma_{AB}  L}{\mu \lambda(1-\lambda)^2\alpha g_2} \geq  \varepsilon_1(d_0^{(N)}) (d^{(N)}_0)^2, \quad N \geq 1.
\end{equation}
It then follows  from the formula (\ref{epid0}) that for $N \geq 1$, 
\begin{equation}\label{CStar}
d_0^{(N)} \leq  \max\Big\{ \frac{4 \sigma_{AB}^{1/3} L^{2/3}}{\mu^{1/3}\lambda^{1/3} (1-\lambda)\alpha^{1/3} g_2^{1/3}}, \frac{4 \sqrt{2} \sigma_{AB}^{1/2} L^{1/2} (1+\lambda)^{1/4}}{\mu^{1/2} \lambda^{1/2}(1-\lambda)\alpha^{1/2} g_2^{1/2}} \Big\}  \stackrel{\rm def.}{=}  C^{\star} .
\end{equation}
Consequently, whenever there is an integer $i^{\star} \leq N$ such that $\sigma_{AB} 2^{i^{\star}+2} (1+\lambda)\gg \mu \lambda (1-\lambda) \alpha g_2 C^{\star}$, the branching construction becomes  self-similar
\begin{align}\label{selfsimilar}
&\varepsilon_{i^{\star}}(d_0^{(N)}) \approx \varepsilon_{i^{\star}+1}(d_0^{(N)}) \approx \ldots \approx \varepsilon_N(d_0^{(N)}) \approx 2\sqrt{\frac{1}{1+\lambda}}, \quad \varepsilon_c(d_0^{(N)}) \approx \sqrt{\frac{2(1-\lambda)}{\lambda}}.
\end{align}
This regime \textit{cannot} be energy minimizing.  Indeed, it follows that $\frac{d_0}{\sigma_{AB}}f_{N+1}(d_0) = \frac{d_0}{\sigma_{AB}}f_{N}(d_0) + O(\frac{\mu(1-\lambda) \alpha g_2 C^{\star}}{2^{N+2} \sigma_{AB}(1+\lambda)}) \approx \frac{d_0}{\sigma_{AB}}f_{N}(d_0)$ for all $d_0\in (0, C^{\star})$ and for all $N \geq i^{\star}$.   In other words, $f_{N+1}(d_0) \approx f_{N}(d_0)$ in this regime.  Consequently, $d_0^{(i^{\star})} \approx d_0^{(i^{\star}+1)} \approx \ldots \approx d_0^{(N)} \approx d_0^{(N+1)} \approx \ldots$, and the energy landscape simplifies.  We observe that  
\begin{equation}
\begin{aligned}
&E^{(N+1)} \approx E^{(N)} + (1-\lambda )\Big(2(1+\lambda)^{1/2} + \lambda \Big)\sigma_{AB}  \quad \text{ if } N \geq i^{\star} \;\; \text{ and } \;\;   2^{i^{\star}} \gg \frac{\mu \lambda(1-\lambda)\alpha g_2 C^{\star}}{4 \sigma_{AB}  (1+\lambda)}
\end{aligned}
\end{equation}
due to three observations: 1).\;the elastic energy for closure domains and the additional branched layer is negligible in this regime, 2).\;the $0-$th layer contribution and surface energy of closure domains is unchanged, and 3).\;the surface energy of branching gets an additional term (i.e., an $N+1$ layer term (\ref{3Result})) that yields the positive contribution as shown given the self-similarity (\ref{selfsimilar}).  Consequently, the optimal number of branching generations belongs to the compact set $N^{\star} \in \{ 0, 1,\ldots, i^{\star}\}$.  Hence, we only need to compute the configurations with energy $E^{(N)}$ for $N \in \{ 0, \ldots, i^{\star}\}$ to be assured of calculating the configuration of optimal energy $E^{(N^{\star})} = E^{\star}$.

As a final remark before summarizing, note that the aspect ratio of the unbranched segment, which is given in terms of all the other segments via (\ref{6Result}), should be positive, i.e., 
\begin{equation}\label{epi0Star}
\varepsilon_0^{\star}  \stackrel{\rm def.}{=} \varepsilon_0\big(\varepsilon_1(d_0^{(N^{\star})}), \ldots, \varepsilon_{N^{\star}}(d_0^{(N^{\star})}), \varepsilon_c(d_0^{(N^{\star})}), d_0^{(N^{\star})}\big) > 0.
\end{equation}
We have essentially ignored this inequality constraint in the optimization above. However, we can easily verify that this inequality holds whenever the branching is energetically preferred over a simple laminate. A related question is whether it is indeed always energetically beneficial for the construction to create such an unbranched layer, i.e., whether $\varepsilon_0^{\star}$ can reach $\pm\infty$  The following proposition fully solves these issues. 

\begin{proposition} 
Let $N^{\star}$, $\varepsilon_0^{\star}$, $L_0^{\star}$, $L_c^{\star}$ and $E^{\star}$ be associated to the energy minimizing configuration, as defined in Section \ref{ssec:optimization}. If $N^{\star} \geq 1$, then $0<\varepsilon^{\star}_0<\infty$.  
\end{proposition}

\begin{proof}
Notice that $\varepsilon^{\star}_0$ and $L^{\star}_0$ are related \emph{via} (\ref{1Result}). In general, $\varepsilon^{\star}_0$ can take values $\varepsilon^{\star}_0\in(-{\infty},+\infty)$, with negative $\varepsilon^{\star}_0$ leading to negative $L_0^{\star}$, and $\varepsilon^{\star}_0=\pm{\infty}$ meaning $L_0^{\star}=0$.  We already excluded the case $\varepsilon^{\star}_0=0$, as the prescribed length $L$ must be a finite number, so it remains to be proved that $\varepsilon^{\star}_0$ is finite and non-negative.
 
For the sake of a contradiction, suppose $\varepsilon^{\star}_0\in\left\{(-\infty,0),+\infty\right\}$, which means $L^{\star}_0\in{}(-\infty,0]$, i.e., the non-branching layer is either absent, or has negative length. Then, the construction within our ansatz may consist of regions with positive energy contribution to $E^{\star}$, and regions with negative energy contributions to $E^{\star}$, such that the total length is equal $L$, as sketched for example in Figs.\ref{epsilon_0}(a,c). As $N^{\star}\geq{1}$, the regions with positive contributions include at least the first branching generation, which means there is a layer of length $L^{\star}_1$ between the $0-$th layer and the closure domains (length $L^{\star}_c$). Regardless of the material parameters and the aspect ratios $\varepsilon_{0}^{\star}$ and $\varepsilon_{1}(d_0^{(N^\star)})$, the energy of the $0-$th unbranched layer per unit length is always lower than the energy of the first branched layer per unit length, as the branched layer includes two times more interfaces and additional elastic strains, i.e.,
\begin{equation}
 \frac{|E^{\star}_0|}{|L^{\star}_0|}<\frac{|E^{\star}_1|}{|L^{\star}_1|} \label{epsilon_ineq}
\end{equation}
whenever $|L_0^{\star}|>0$. Here, $E^{\star}_0$ is calculated from $L_0^{\star}$ using (\ref{1Result}), and $E^{\star}_1$ is the sum of the surface energy and elastic energy of the $1-$st branching layer for optimized $\varepsilon^{\star}_1$. We discuss now two separate cases:
\begin{itemize}
\item{}$|L^{\star}_0|{\leq}L^{\star}_1$, (Fig.~\ref{epsilon_0}(a)). Due to (\ref{epsilon_ineq}), replacing the  first branched layer by an unbranched layer of length $L_1^{\star}-|L_0^{\star}|$ always reduces energy, and thus, the construction shown in (Fig.~\ref{epsilon_0}(b)) has the same total length $L$ but lower energy, which contradicts the optimality of $N^{\star}$, $\varepsilon_0^{\star}$, $L_0^{\star}$, $L_c^{\star}$ and $E^{\star}$. The same reasoning can be obviously used also for $L^{\star}_0=0$ (i.e., $\varepsilon_0^{\star}=\pm\infty$) and  $L^{\star}_0=L^{\star}_1$.
\item{}$|L^{\star}_0|{>}L^{\star}_1$, (Fig.~\ref{epsilon_0}(c)). Due to (\ref{epsilon_ineq}), removing completely the first branched layer and making the length of the $0-$th layer shorter by $L_1^{\star}$ is an energy reducing operation. Also making $d_0$ two times smaller reduces the energy of the unbrached layer for $L_0<0$. Hence, the construction in Fig.~\ref{epsilon_0}(d) has, again, lower energy than the one in Fig.~\ref{epsilon_0}(c), and the proof is completed.  
\end{itemize}
\end{proof}

\begin{figure}[htb]
\centering
 \includegraphics[width=\textwidth]{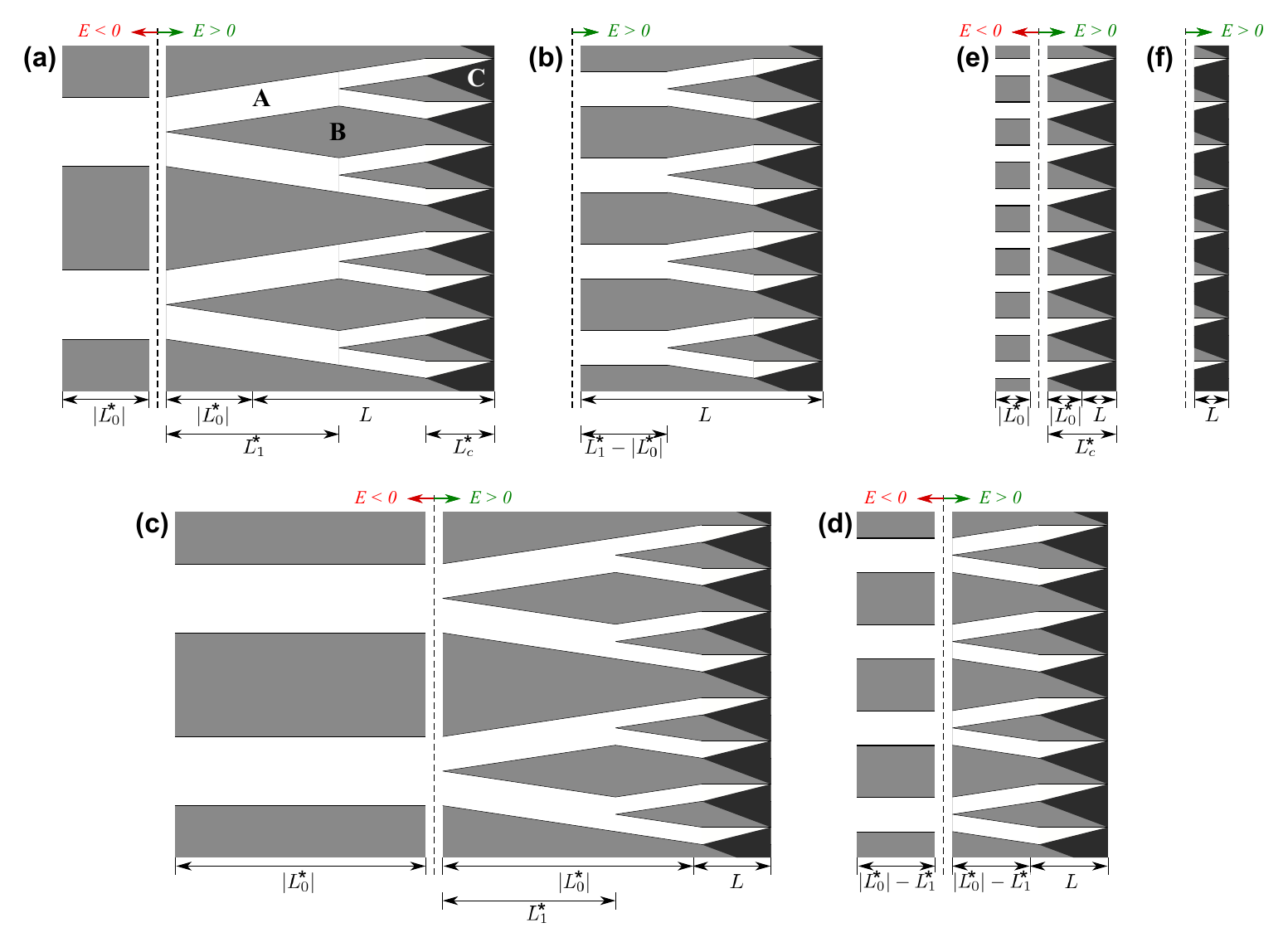}
 \caption{On the energy minimizing configuration and the constraint on $\varepsilon_0^{\star}$: (a,c,e) possible constructions with $\varepsilon_0^\star<0$; (b,d,f) alternative construction with the same total lengths $L$ as in (a,c,e) but lower energies. The dashed line always separates the parts of the construction with a negative energy from those with a positive energy.}\label{epsilon_0}
\end{figure}

For completeness, let us also briefly discuss the case when $N^{\star}=0$. The construction, in this case, consists only of the unbranched layer (a simple laminate) and closure domains. If the surface energy is very high, or the length $L$ is very small, the energy-minimizing configuration within our ansatz can be indeed the one sketched in Fig.~\ref{epsilon_0}(e), with negative $L_0^{\star}$. However, it is also directly seen that there are other constructions beyond our ansatz that have lower energy; for example, the one in Fig.~\ref{epsilon_0}(f). This construction is perfectly valid as it satisfies the required boundary conditions and compatibility conditions at all interfaces.  Thus, the possible result $N^{\star}=0$ and $\varepsilon_0^{\star}<0$ indicates that our ansatz fails to give physically realistic construction, but it also call into question `twinning' as a mechanism to minimize energy for the given experimental parameters. Consequently, it seems an unlikely result in the physically relevant regime of experimental parameters.

\subsection{A numerical procedure for optimality}\label{ssec:numProc}
The optimization above provides a procedure to numerically compute the energy minimizing microstructure given the experimental parameters.  To summarize:
\begin{itemize}
\item \textit{Check if laminate microstructure is energy minimizing $-$} Solve for $d_0^{(0)}$ in (\ref{d0N}) and check if $\varepsilon_0^{simp.} \stackrel{\rm def.}{=}  \varepsilon_0(\varepsilon_c(d_0^{(0)}), d_0^{(0)})$ is positive (see (\ref{6Result}) and (\ref{epcd0}) for the evaluation).  If it is, then proceed.  If not, then $AB$ laminate microstructure is not energy minimizing for the given parameters.
\item \textit{The number of branching generations to evaluate $-$} Compute $C^{\star}$ in (\ref{CStar}), and choose an positive integer $i^{\star}$ such that $2^{i^{\star}} \gg  \frac{\mu \lambda(1-\lambda)\alpha g_2 C^{\star}}{4 \sigma_{AB}  (1+\lambda)}$.
\item \textit{The energy landscape $-$} Solve for $d_0^{(N)}$ in (\ref{d0N}) and evaluate the energy $E^{(N)}$ in (\ref{EN})  for each $N \in \{ 0,1,\ldots, i^{\star}\}$.  
\item \textit{The optimal number of branching generations $-$} This is given by the $N^{\star} \in \{ 0,1,\ldots, i^{\star}\}$, which gives the minimal energy for the evaluation above.
\item \textit{The optimal twin width of the unbranched segment $-$} This is given by $d_0^{\star}  \stackrel{\rm def.}{=} d_0^{(N^{\star})}$.
\item \textit{The optimal lengths of each segment of the construction $-$}  The unbranched segment length is $L_0^{\star}  \stackrel{\rm def.}{=} \tfrac{(1-\lambda) \alpha d_0^{\star}}{ \varepsilon_0^{\star}}$ for $\varepsilon_0^{\star}$ in (\ref{epi0Star}), the branched segment lengths are $L_i^{\star} \stackrel{\rm def.}{=} \tfrac{(1-\lambda) \alpha d_0^{\star}}{ 2^{i} \varepsilon_i(d_0^{\star})}$ for $i \in \{1,\ldots, N^{\star}\}$ and for $\varepsilon_i(d_0^{\star})$  evaluated via (\ref{epid0}) or (\ref{epNd0}), and the closure domain length is $L_c^{\star} \stackrel{\rm def.}{=} \tfrac{(1-\lambda) \alpha d_0^{\star}}{ 2^{N^{\star}} \varepsilon_c(d_0^{\star})}$ for $\varepsilon_c(d_0^{\star})$ in (\ref{epcd0}).
\item \textit{The total energy $-$}  This is given by $E^{\star}  = E^{(N^{\star})}$, and it satisfies the minimality condition
\begin{equation}\label{infEstar}
E^{\star} = \inf_{N \in \{0, 1, \ldots\}}\Big\{ \inf_{\varepsilon_1,\ldots, \varepsilon_N, \varepsilon_c, d_0 > 0} E(\varepsilon_1, \ldots, \varepsilon_N, \varepsilon_c, d_0) \Big\}.
\end{equation}.  
\end{itemize}

\subsection{Energy scaling for the construction $-$ an upper bound}

 We now discuss the energy scaling of the detailed construction, i.e., the construction with the unbranched $0-$th layer, finite $N$ and closure domains. Our aim is to show that $E^{\star}$ follows the $E^{\star}\sim \mu^{1/3} \sigma_{AB}^{2/3} L^{1/3}$ scaling in the $\sigma_{AB}\rightarrow{}0$ limit; thus ensuring that our construction is consistent with the scaling argument in Section \ref{scalingargument}. 
It is not easy to argue this directly. Instead, we bound the energy $E^{\star}$ from above in this section, and, in the next section, prove an ansatz-free lower bound on the energy (under mild additional assumptions). The two bounds give the same scaling for sufficiently small (i.e., physically relevant) $\sigma_{AB}$, which means that also $E^{\star}$ follows this scaling at least in the considered limit.  

By this approach, we prove not only that $E^{\star}$ scales as $\mu^{1/3} \sigma_{AB}^{2/3} L^{1/3}$ for $\sigma_{AB}\rightarrow{0}$, but also that the energy of any even more detailed or more physically realistic construction (i.e., constructions with energy lower than $E^{\star}$) would adopt this scaling in the considered limit. Hence, it is justified to assume that the energy of the real microstructures appearing at the habit planes should also depend on the elastic constants, the surface energy and the length of the laminate following a similar scaling law, when the surface energy goes to zero.
 
For the upper bound, we replace the configuration of energy $E^{\star}$ with a configuration of higher energy $\hat{E}$ but simpler structure.  We then estimate the energy $\hat{E}$ from above to obtain a suitable bound.  As we will be using estimates in this section, we take $C >0$ to represent a constant that can depend on the crystallographic parameters in (\ref{compat1}-\ref{compat2}), but not on $\sigma_{AB}$, $\mu$ or $L$.  This constant can change from line to line.  
 
We assume the experimental parameters are such that laminate $AB$ microstructure is energy minimizing, in the sense that $\varepsilon_{0}^{simp.}$ from the numerical procedure in Section \ref{ssec:numProc} is positive.  This gives a physical configuration with energy $E^{\star}$.  We replace this configuration with one of higher energy via the following construction similar to Chan and Conti's \citep{ChanConti2}: We introduce a parameter $\theta \in (1/4,1/2)$ and optimize over twin widths $d_0 \in (0, \tfrac{L(1-\theta)}{(1-\lambda)\alpha}]$ under the construction for branching generations and aspect ratios
\begin{equation}
\begin{aligned}\label{hatDefs}
&\hat{N}(d_0)  \stackrel{\rm def.}{=} \max  \Big\{ i  \in \mathbb{N} \colon \frac{(1-\lambda) \alpha d_0}{2^i (1-\theta)\theta^i L } \leq 1  \Big\} ,\\
&\hat{\varepsilon}_i(d_0)  \stackrel{\rm def.}{=} \frac{(1-\lambda) \alpha d_0}{(1-\theta)2^i \theta^i L}  , \quad i = 0,\ldots, \hat{N}(d_0), \quad \hat{\varepsilon}_c(d_0) \stackrel{\rm def.}{=} \frac{(1-\lambda) \alpha d_0}{2^{\hat{N}(d_0)} \theta^{\hat{N}(d_0)+1} L}.
\end{aligned}
\end{equation}
The layer lengths for this self-similar construction are therefore  $L_i = (1-\theta) \theta^i L$ and $L_c = \theta^{\hat{N}(d_0)+1} L$, guaranteeing that $\sum_{i = 0}^{\hat{N}(d_0)} L_i + L_c = L$.  This implies that the constraint on the aspect ratio in (\ref{6Result}) holds trivially, i.e.,  $\hat{\varepsilon}_0(d_0) = \varepsilon_0(\hat{\varepsilon}_1(d_0), \ldots, \hat{\varepsilon}_{\hat{N}(d_0)}(d_0), \hat{\varepsilon}_c(d_0))$. Consequently, this family of configurations is consistent with our ansatz for branching microstructure, and it has energy
\begin{equation}
\begin{aligned}\label{hatEd0}
\hat{E}(d_0)  &\stackrel{\rm def.}{=} E(\hat{\varepsilon}_1(d_0), \ldots, \hat{\varepsilon}_{\hat{N}(d_0)}(d_0), \hat{\varepsilon}_c(d_0), d_0) \\
&\stackrel{\rm def.}{=} \hat{E}_0(d_0)  + \hat{E}_{elast.}^{br.}(d_0) + \hat{E}_{surf.}^{br.}(d_0) + \hat{E}_{elast.}^{cl.}(d_0) + \hat{E}_{surf.}^{cl.}(d_0)
\end{aligned}
\end{equation}
as derived in Section \ref{ssec:derivation}.  

To estimate the different terms in the energy (\ref{hatEd0}), it is useful to notice that the definitions in (\ref{hatDefs}) furnish the aspect ratio estimates and identity 
\begin{equation}\label{ests}
\begin{aligned}
&0< \hat{\varepsilon}_{i}(d_0) \leq 1 \quad \text{ for } \quad i = 1,\ldots, \hat{N}(d_0) -1, \\
&2\theta <  \hat{\varepsilon}_{\hat{N}(d_0)}(d_0) \leq 1 \quad \text{ and } \quad  \frac{\theta}{1-\theta} \hat{\varepsilon}_c(d_0) = \hat{\varepsilon}_{\hat{N}(d_0)}(d_0).
\end{aligned}
\end{equation}
Thus, in collecting the surface energy terms (\ref{1Result}), (\ref{3Result}) and (\ref{5Result}), we observe that each layer is the sum of at most three terms: an $O(\hat{\varepsilon}_{i}(d_0)^{-1})$ term, an $O(1)$ term and an $O(\hat{\varepsilon}_i(d_0))$ term\textemdash all proportional to $\sigma_{AB} = \sigma |\mathbf{a}|$.  The $O(\hat{\varepsilon}_{i}(d_0)^{-1})$ terms dominate since the aspect ratios are bounded above by unity (\ref{ests}). It therefore follows that
\begin{equation}
\begin{aligned}
\hat{E}_0(d_0) + \hat{E}_{surf.}^{br.}(d_0) + \hat{E}_{surf.}^{cl.}(d_0) \leq C\sigma_{AB} \Big(\sum_{i = 0}^{\hat{N}(d_0)} \frac{1}{\hat{\varepsilon}_i(d_0)} \Big) \leq C\sigma_{AB} \Big(\sum_{i = 0}^{\hat{N}(d_0)}  2^{i} \theta^{i} \Big)\frac{L}{d_0}.
\end{aligned}
\end{equation}
In addition, as $\hat{\varepsilon}_{\hat{N}(d_0)}(d_0)$ is bounded away from zero and infinity (given (\ref{ests})), we can bound the terms in $\hat{E}_{elast.}^{cl.}(d_0)$ by a quantity proportional to $\hat{\varepsilon}_{\hat{N}(d_0)}(d_0)$ and conclude that 
\begin{equation}
\begin{aligned}
&\hat{E}_{elast.}^{br.}(d_0)  + \hat{E}_{elast.}^{cl.}(d_0) \leq C \mu  d_0 \Big(\sum_{i=0}^{\hat{N}(d_0)} \frac{\hat{\varepsilon}_i(d_0)}{2^{i}} \Big) \leq C \mu \Big( \sum_{i=0}^{\hat{N}(d_0)} \frac{1}{4^{i} \theta^i} \Big) \frac{d_0^2}{L}.\\
\end{aligned}
\end{equation}
Since $\theta \in (1/4,1/2)$, the two series in these estimates converge if we replace $\hat{N}(d_0)$ with $\infty$, yielding a further upper bound estimate of these energies.   

By combining all these estimates and substituting into (\ref{hatEd0}), we obtain the overall estimate  for this family of constructions
\begin{equation}
\begin{aligned}\label{Ehatd0Est}
\hat{E}(d_0) \leq C \Big(\sigma_{AB} \frac{L}{d_0} +  \mu \frac{d_0^2}{L} \Big).
\end{aligned}
\end{equation}
In taking the infimum of both sides of (\ref{Ehatd0Est}) amongst twin widths $d_0 \in (0, \tfrac{L(1-\theta)}{(1-\lambda)\alpha}]$, we then obtain
\begin{equation}
\begin{aligned}
\hat{E}\stackrel{\rm def.}{=}  \inf \Big\{ \hat{E}(d_0) \colon d_0 \in (0, \tfrac{L(1-\theta)}{(1-\lambda)\alpha}]\Big\}  \leq C \min \Big\{ \mu^{1/3} \sigma_{AB}^{2/3} L^{1/3}, \mu L \Big\} .
\end{aligned}
\end{equation}
Finally, $E^{\star} \leq \hat{E}$ due to (\ref{infEstar}).  So we conclude that the optimal configuration within our ansatz has energy  
\begin{equation}
\begin{aligned}\label{keyUpper}
E^{\star} \leq C(\mathbf{a}, \mathbf{m}, \mathbf{n}, \lambda, \alpha)  \min \Big\{ \mu^{1/3} \sigma_{AB}^{2/3} L^{1/3}, \mu L \Big\} .
\end{aligned}
\end{equation}
As this is the key result, we emphasize here that the constant $C >0$ depends on the crystallographic parameters. 

The latter term in the upper bound is the minimizer only in the case that $\mu L \leq \sigma_{AB}$.  This inequality is unlikely to hold for typical shape memory alloys, and so generically this upper bound suggests the desired scaling $E^{\star} \leq C\mu^{1/3} \sigma_{AB}^{2/3} L^{1/3}$. Note that $E^{\star}$ is an upper bound to the energy of the real branched microstructure $E$, and thus, $E \leq C\mu^{1/3} \sigma_{AB}^{2/3} L^{1/3}$.

\subsection{Energy scaling and optimality of the construction $-$ a lower bound}\label{ssec:LBSec}
We now turn to an ansatz-free lower bound on the energy that, combined with the upper bound previously, gives the energy scaling $\sim \mu^{1/3} \sigma_{AB}^{2/3}L^{1/3}$ for branching microstructure in the physically relevant regime of parameters.  The starting point here is the study of the elastic energy (\ref{energy}) after a convenient normalization and with some additional assumptions on the reference configuration, deformation and the structure of the elastic energy density.  Precisely, we study the elastic energy 
\begin{equation}\label{energyLB}
\mathcal{E}^{\hat{\sigma}}(\mathbf{y}, \Omega_{L,H}) \stackrel{\rm def.}{=} \int_{\Omega_{L,H}}\Big( \hat{\varphi}(\nabla \mathbf{y}) + \hat{\sigma} |\nabla^2 \mathbf{y}| \Big) dx
\end{equation}
under the following hypotheses: 
\begin{itemize}
\item \textit{The reference configuration $-$} $\Omega_{L,H} \subset \mathbb{R}^3$ is a parallelepiped domain of square cross-section $H^2$ and length $L$ parallel to the habit plane normal $\mathbf{m}$ (Fig.~\ref{fig:ChangeVar} top-left).  Without loss of generality, we take $\mathbf{m} = \mathbf{e}_1$ and $\mathbf{m}^{\perp} = \mathbf{e}_2$ for $\{\mathbf{e}_1, \mathbf{e}_2, \mathbf{e}_3\}$ the standard basis on $\mathbb{R}^3$.
\item \textit{A normalization $-$} The energy density $\hat{\varphi}(\mathbf{F}) \stackrel{\rm def.}{=} 2\varphi(\mathbf{F})/\mu$ and interfacial parameter $\hat{\sigma} \stackrel{\rm def.}{=} 2\sigma_{AB}/(|\mathbf{a}| \mu)$ have been normalized by the characteristic modulus $\mu$ of the martensite phase.
\item \textit{The crystallographic theory $-$} The compatibility conditions (\ref{compat1}-\ref{compat2}) hold.  In addition, we assume that the crystallographic parameters satisfy the condition\footnote{This is a benign condition.  Conventional shape memory alloys are nearly volume preserving, and the transformations $\mathbf{A}, \mathbf{B}$ are $\approx \mathbf{I}$.  Consequently, $\mathbf{a}$ should be nearly parallel to $\mathbf{n}^{\perp}$ (see for example (\ref{crys1})) and $\mathbf{I} + \mathbf{b} \otimes \mathbf{m} = \lambda \mathbf{A} + (1-\lambda) \mathbf{B} \approx \mathbf{Id}$.  This yields the result that $|(\mathbf{I} + \mathbf{b} \otimes \mathbf{m}) \mathbf{n}^{\perp} \cdot \mathbf{a}|$ is to leading order $\approx |\mathbf{a}|$ in these alloys, which is far from zero.} $(\mathbf{Id} + \mathbf{b} \otimes  \mathbf{m}) \mathbf{n}^{\perp} \cdot \mathbf{a} \neq 0$ for $\mathbf{n}^{\perp} \stackrel{\rm def.}{=} -(\mathbf{n} \cdot \mathbf{e}_2) \mathbf{e}_1 + (\mathbf{n} \cdot \mathbf{e}_1) \mathbf{e}_2$.  For future reference, we also define the direction $\boldsymbol{\nu}  \stackrel{\rm def.}{=} \frac{(\mathbf{Id} + \mathbf{b} \otimes \mathbf{m}) \mathbf{n}^{\perp} }{|(\mathbf{Id} + \mathbf{b} \otimes \mathbf{m}) \mathbf{n}^{\perp}|}$.  This is well-defined due to the added hypothesis here.
\item \textit{A two-well structure $-$} We fix a temperature \textit{below} the transition temperature and assume the energy density $\hat{\varphi}(\mathbf{F})$ is minimized and equal to zero on $K  \stackrel{\rm def.}{=}SO(3) \mathbf{A} \cup SO(3) \mathbf{B}$ and satisfies $\hat{\varphi}(\mathbf{F}) \geq \text{dist}^2(\mathbf{F},K)$.  
\item \textit{Boundary conditions on the deformation $-$} We study the energy (\ref{energyLB}) subject to continuous deformations with bounded energy (i.e., $\mathbf{y} \in C(\overline{\Omega}_{L,H}, \mathbb{R}^3) \cap W^{2,1}(\Omega_{L,H},\mathbb{R}^3)$)\footnote{It is possible to relax the interfacial energy studied here to include deformation gradients $\nabla \mathbf{y}$ of \textit{bounded variation}, i.e., those deformations gradients consistent with the formalism of our construction (\ref{assump1}).  However, rather than introduce the additional mathematical machinery of BV-functions, we simply note that the lower bound \textit{does not} change under this relaxation.}.  In addition, we assume that these deformations satisfy the boundary conditions on the left and right boundary: $\mathbf{y}(\mathbf{x}) = (\mathbf{I} + \mathbf{b} \otimes \mathbf{m}) \mathbf{x}$ for $\mathbf{x} \cdot \mathbf{m} \in \{0, L\}$.  We call this space of deformations $\mathcal{M}$.  
\end{itemize}

We show rigorously in the next subsection \ref{ssec:LowerBoundProof} that the infimum of this energy satisfies
\begin{align}\label{mainEst}
\mathcal{E}^{\star}(\hat{\sigma},L,H)  \stackrel{\rm def.}{=}  \inf \Big\{ \mathcal{E}^{\hat{\sigma}}(\mathbf{y}, \Omega_{L,H}) \colon \mathbf{y} \in \mathcal{M} \Big\}  \geq c H^3 \min  \Big\{ \frac{\hat{\sigma}^{2/3} L^{1/3}}{H}, \frac{H}{L},  \frac{L}{H} \Big\} 
\end{align}
for some constant $c \equiv c(\mathbf{a}, \mathbf{b}, \mathbf{n}, \mathbf{m}, \lambda, \alpha) > 0$ that depends only on the crystallographic parameters.  The proof of this result is based in an essential way on the results of \cite{ChanConti2}, who studied an analogous two dimensional problem with two energy wells and boundary conditions on all sides.  Thus, as it relates to actual microstructure in bulk shape memory alloys, there are limitations to this result that merit discussion. 

 In particular, owing to the symmetry transformation from austenite to martensite, most shape memory alloys have more than two martensitic wells below the transition temperature.  However, our proof relies on the fact that we consider two and only two wells.  If one were to actually approach the lower bound using the correct number of energy wells, then one would have to account for the possibility that the additional well(s) may lead to microstructure that reduces the overall energy\textemdash possibly to the point where the estimate (\ref{mainEst}) no longer applies.  In addition, our proof relies on the fact that we have boundary conditions (i.e., the austenite meeting the martensite) on two sides.  In contrast, typical nucleation of martensite occurs as propagation of a ``front".  The martensite takes advantage of free-surfaces\textemdash at, say, corners, defects, or grain boundaries\textemdash to propagate in a manner that usually avoids more than one austenite-martensite interface for the transforming band.  For these two reasons, our result here is far from the definitive characterization of the lower bound energy of the austenite-martensite interface in bulk shape memory alloys.  Nevertheless, we can use this estimate to derive that the optimal construction (Section \ref{ssec:numProc}) satisfied $E^{\star} \sim \mu^{1/3} \sigma_{AB}^{2/3} L^{1/3}$ in the $\sigma_{AB}\rightarrow{0}$ limit. 
 
In this direction, notice that the elastic energy density $\varphi(\mathbf{F})$ in (\ref{assump2}), for which we developed this construction, is a linearized variant of an elastic energy density $\frac{\mu}{2} \text{dist}^2(\mathbf{F} ,K)$ under the assumption $\mathbf{A}, \mathbf{B} \approx \mathbf{I}$.   While this energy density is not technically a $\frac{\mu}{2} \hat{\varphi}(\mathbf{F})$ satisfying the above hypothesis, the distinction is for a range of transformations $\mathbf{A},\mathbf{B}$ and perturbations $\boldsymbol{\delta} \mathbf{A}, \boldsymbol{\delta} \mathbf{B}$ that are  atypical of shape memory alloys.  We treat the typical case below.
 
 The construction giving energy $E^{\star}$ is of length $L$ (parallel to the habit plane normal $\mathbf{m}$) and unit depth and width.  It is compatible with the austenite on one side, and, according to (\ref{eq:moved}), satisfying the macroscopic compatibility conditions everywhere in its interior.  Thus, by taking this configuration, adding an extra unbranched layer of length $L$ and repeating the branching construction (using the same branching segments and the same optimized parameters) in the reverse direction in this additional layer, we obtain a configuration of length $2L$ (parallel to $\mathbf{m}$)  and with energy $2 E^{\star}$.  This configuration satisfies the boundary conditions consistent with the hypotheses for the above lower bound (i.e., this modification is now compatible with the austenite on both sides).  Thus, we apply the estimate (\ref{mainEst}) to the energy $E^{\star}$ for \textit{typical} shape memory alloys to obtain 
\begin{align}\label{LBImport}
H^2 E^{\star} \geq \frac{\mu}{4}  \mathcal{E}^{\star}\Big(\frac{2\sigma_{AB}}{\mu |\mathbf{a}|},L,H\Big)  \geq  c H^3 \min \Big\{ \frac{\mu^{1/3} \sigma_{AB}^{2/3} L^{1/3}}{H}, \mu \frac{H}{L}, \mu \frac{L}{H} \Big\} 
\end{align}
(recalling the normalization for $\hat{\sigma}$).  Here, the first inequality is since the minimization for $\mathcal{E}^{\star}$ is ansatz-free whereas the minimization for $E^{\star}$ is not.  In addition, the constant in the lower bound depends only on the crystallographic parameters $c \equiv c(\mathbf{a}, \mathbf{b}, \mathbf{n}, \mathbf{m}, \lambda, \alpha) > 0$.

Finally, in typical shape memory alloys, we expect the non-dimensionalized parameter $\frac{\sigma_{AB}}{\mu L}$ to fall between a coarse lower and upper bound of $10^{-11}$ and  $10^{-5}$ (see Tab.1).  We further expect the aspect ratio $H/L$ to not deviate significantly from being $O(1)$. In particular, we expect the inequality $\big(\frac{\sigma_{AB}}{\mu L} \big)^{2/3} \leq \min \big\{1, \frac{H^2}{L^2} \big\}$ to hold trivially in these materials.   As a consequence of these expectations, minimization of the upper and lower bounds in (\ref{keyUpper}) and (\ref{LBImport}) yields the scaling result
\begin{align}
c \mu^{1/3} \sigma_{AB}^{2/3} L^{1/3} \leq  E^{\star} \leq C \mu^{1/3} \sigma_{AB}^{2/3} L^{1/3}, 
\end{align}
since $0<c\leq C$ are constants depending only on the crystallographic parameters. Hence, $E^{\star}$ indeed follows the desired scaling in the $\sigma_{AB}\rightarrow{}0$ limit.

\subsection{The proof of the lower bound}

\label{ssec:LowerBoundProof}
\begin{figure}[t!]
\centering
\includegraphics[width = 3.8in]{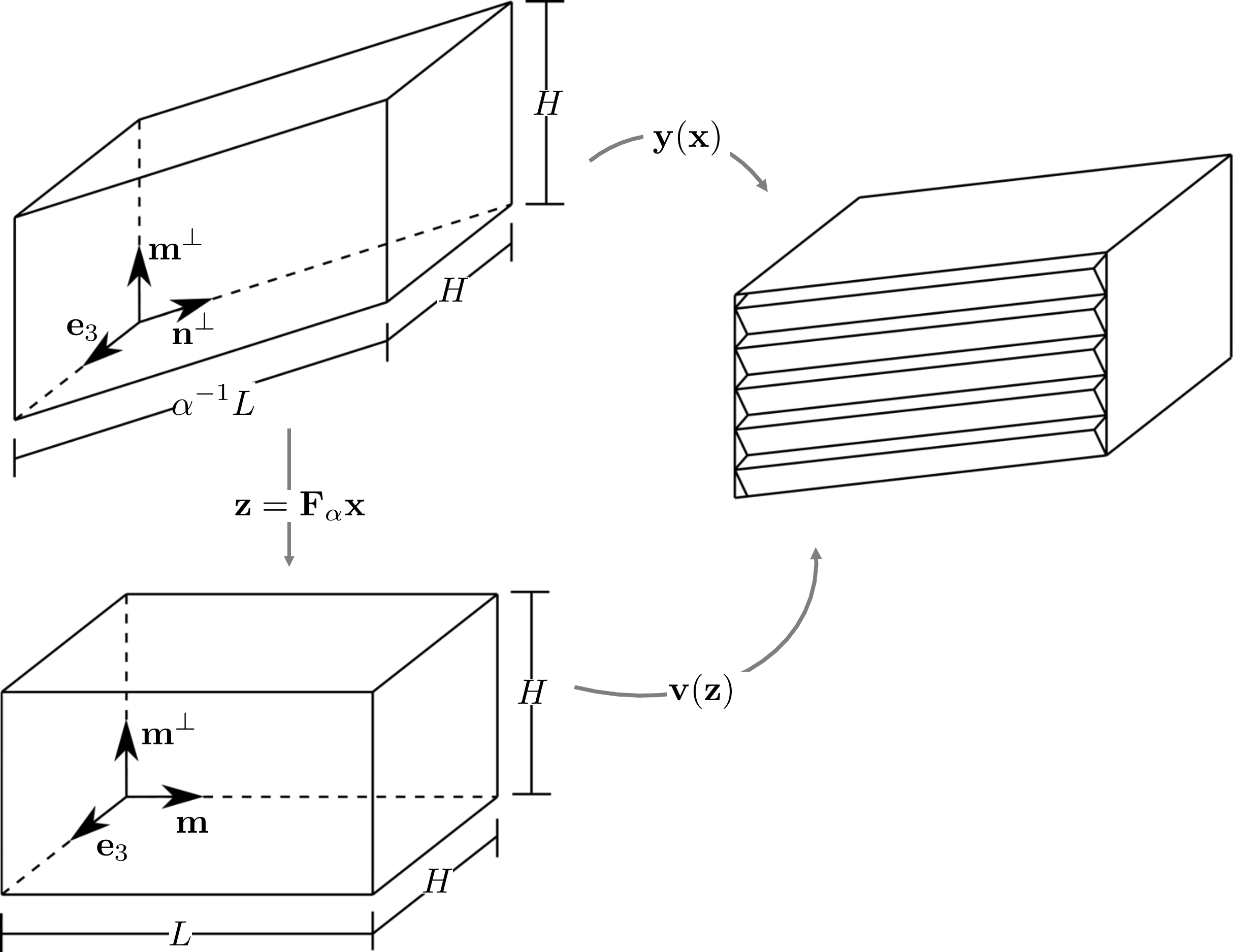}
\caption{A change of variables which allows us to consider the energy minimization problem $\mathcal{E}^{\star}(\hat{\sigma}, L,H)$ amongst deformations which map from a rectangular prism reference configuration.}
\label{fig:ChangeVar}
\end{figure}

\subsubsection{Formulation}  

In this final part of the Appendix, we prove the lower bound (\ref{mainEst}) for the energy $\mathcal{E}^{\hat{\sigma}}(\mathbf{y}, \Omega_{L,H})$ in (\ref{energyLB}) under the stated hypotheses in Section \ref{ssec:LBSec}.

In what follows, we find it convenient to reformulate the energy minimization problem $\mathcal{E}^{\star}(\hat{\sigma}, L, H)$ so that the reference configuration is a rectangular prism as show in the lower-left of Fig.~\ref{fig:ChangeVar}.  Observe that by a uniform simple shear of $\Omega_{L,H}$, we can achieve the rectangular prism $\mathcal{R}_{L,H}$ depicted in the figure.  The simple shear is given by 
\begin{equation}
\begin{aligned}\label{eq:tFDef}
\mathbf{F}_{\alpha} \stackrel{\rm def.}{=}  \mathbf{Id} -  \alpha^{-1} \big(\mathbf{n} \cdot \mathbf{m}\big) \mathbf{m}^{\perp} \otimes \mathbf{m}, \quad \text{ with } \quad \mathbf{F}_{\alpha}^{-1} =  \mathbf{Id} + \alpha^{-1} \big(\mathbf{n} \cdot \mathbf{m}\big) \mathbf{m}^{\perp} \otimes \mathbf{m}.
\end{aligned}
\end{equation} 
Thus, we let $\mathbf{z} \stackrel{\rm def.}{=}  \mathbf{F}_{\alpha} \mathbf{x}$, and we associate to any deformation $\mathbf{y} \colon \Omega_{L,H} \rightarrow \mathbb{R}^3$ a function $\mathbf{v} \colon \mathcal{R}_{L,H} \rightarrow \mathbb{R}^3$ defined by $\mathbf{v}(\mathbf{z}(\mathbf{x})) \stackrel{\rm def.}{=} \mathbf{y}(\mathbf{x})$ for $\mathbf{x} \in \Omega_{L,H}$.  Since $\det \mathbf{F}_{\alpha} = 1$, this change of variables yields the identity  
\begin{equation}
\begin{aligned}
\mathcal{E}^{\hat{\sigma}}(\mathbf{y}, \Omega_{L,H}) = \tilde{\mathcal{E}}_{el}(\mathbf{v}, \mathcal{R}_{L,H}) + \tilde{\mathcal{E}}_{int}^{\hat{\sigma}}(\mathbf{v},\mathcal{R}_{L,H}) \stackrel{\rm def.}{=}  \tilde{\mathcal{E}}^{\hat{\sigma}} (\mathbf{v}, \mathcal{R}_{L,H}),
\end{aligned}
\end{equation}
where the elastic part and interfacial part, after the change of variables, take the forms
\begin{equation}
\begin{aligned}\label{eq:EintBound}
&\tilde{\mathcal{E}}_{el}(\mathbf{v}, \mathcal{R}_L) \stackrel{\rm def.}{=}  \int_{\mathcal{R}_{L,H}} \hat{\varphi}\big( (\nabla \mathbf{v}) \mathbf{F}_{\alpha}) dz, \\
&\tilde{\mathcal{E}}_{int}^{\sigma}(\mathbf{v}, \mathcal{R}_L) \stackrel{\rm def.}{=} \hat{\sigma} \int_{\mathcal{R}_{L,H}}  \sqrt{\Big( v_{i,lm} (F_{\alpha})_{lj} (F_{\alpha})_{mk} \Big) \Big( v_{i,l'm'} (F_{\alpha})_{l'j}(F_{\alpha})_{m'k} \Big) } dz   \geq \hat{\sigma} c \alpha^2 \int_{\mathcal{R}_{L,H}} |\nabla^2 \mathbf{v}| dz, 
\end{aligned}
\end{equation}
respectively.  For the latter, the identity uses index notation with repeated indices summed, and the constant $c > 0$ in the lower bound is universal.  A proof of this lower bound is provided in Proposition \ref{appendProp} below.  

Hence, the variational problem $\mathcal{E}^{\star}(\hat{\sigma}, L, H)$ in (\ref{mainEst}) can be reformulated by studying the energy $\tilde{\mathcal{E}}^{\hat{\sigma}}(\mathbf{v}, \mathcal{R}_{L,H})$ subject to functions of the form 
\begin{equation}
\begin{aligned}
\mathcal{M}_{\alpha} &\stackrel{\rm def.}{=}  \Big\{ \mathbf{u} \in W^{2,1}(\mathcal{R}_{L,H}, \mathbb{R}^3) \cap C(\overline{\mathcal{R}}_{L,H}, \mathbb{R}^3) \colon \mathbf{u}(\mathbf{z}) = (\mathbf{Id} + \mathbf{b} \otimes \mathbf{m}) \mathbf{F}_{\alpha}^{-1} \mathbf{z} ,\;\; \mathbf{z} \cdot \mathbf{m} \in \{ 0, L \} \Big\} .
\end{aligned}
\end{equation}
Specifically,  $\mathcal{E}^{\star}(\hat{\sigma}, L, H)$ is given equivalently by variational problem
\begin{equation}
\begin{aligned}
\mathcal{E}^{\star}(\hat{\sigma}, L, H)= \inf \Big\{ \tilde{\mathcal{E}}^{\hat{\sigma}}(\mathbf{v}, \mathcal{R}_{L,H}) \colon \mathbf{v} \in \mathcal{M}_{\alpha} \Big\} .
\end{aligned}
\end{equation}

\subsubsection{Main result on the lower bound.}
The lower bound energy estimate in (\ref{mainEst}) is obtained as a consequence of the following theorem:
\begin{theorem}\label{lbTheorem}
Assume the energy $\mathcal{E}^{\hat{\sigma}}(\mathbf{y}, \Omega_{L,H})$ in (\ref{energyLB}) under the stated hypotheses in Section \ref{ssec:LBSec}.  Then, there exists a universal constant $c_{\star}> 0$ such that the minimal energy $\mathcal{E}^{\star}(\hat{\sigma}, L, H)$ satisfies
\begin{equation}
\begin{aligned}\label{eq:lbBound}
&\mathcal{E}^{\star}(\hat{\sigma}, L, H) \geq c_{\star} H^3 \min \Big\{   \bar{c}_{K}  \frac{\hat{\sigma}^{2/3} L^{1/3}}{H},  \hat{c}_{K}  \frac{H}{L}, \hat{c}_{K} \frac{L}{H} \Big\} \\
\end{aligned}
\end{equation}
for $ \bar{c}_{K} \stackrel{\rm def.}{=}  \alpha^{10/3} c_{K}^{4/3}$,  $\hat{c}_{K} \stackrel{\rm def.}{=}  \alpha^2 c_{K}^{2}$ and for $c_{K}$ depending on the crystallographic parameters via 
\begin{equation}
\begin{aligned}\label{eq:cKDef}
c_{K} \stackrel{\rm def.}{=}   \min_{\mathbf{G} \in K} \max \Big\{ \alpha \big|\boldsymbol{\nu} \cdot \mathbf{G} \mathbf{n}^{\perp} - |(\mathbf{Id} + \mathbf{b} \otimes \mathbf{m}) \mathbf{n}^{\perp}|\big|, \big|\boldsymbol{\nu} \cdot (\mathbf{G} - \mathbf{Id}) \mathbf{m}^{\perp} \big|, \big|\boldsymbol{\nu} \cdot (\mathbf{G} - \mathbf{Id} ) \mathbf{e}_3\big|\Big\} > 0.
\end{aligned}
\end{equation}
\end{theorem}
\begin{remark}[On the lower bound parameters]
The constant $c_{K}$ is actually $>0$ if and only if the crystallographic parameters satisfy $(\mathbf{Id} + \mathbf{b} \otimes  \mathbf{m}) \mathbf{n}^{\perp} \cdot \mathbf{a} \neq 0$.  Consequently, this hypothesis is required for the $\hat{\sigma}^{2/3} L^{1/3}$ scaling using our strategy of proof. 
\end{remark}
The theorem follows from a series of propositions.  We state the propositions below, and reserve proof for the coming section.   The first observation to make is that we can always isolate a strip $\mathcal{S}_{L,\delta} \stackrel{\rm def.}{=}  (0,L) \times (s', s'+\delta) \times (s'', s'' + \delta) \subset \mathcal{R}_{L,H}$ and a cube inside that strip $\mathcal{Q}_{\delta} \stackrel{\rm def.}{=}  (s, s+ \delta) \times (s', s'+\delta) \times (s'', s'' + \delta) \subset \mathcal{S}_{L,\delta}$ such the energy on these domains is no higher than the average energy of the entire system:
\begin{proposition}\label{firstProp}
Let $\mathbf{v} \in \mathcal{M}_{\alpha}$.  For any  $\delta \in (0, \min\{H, L\}]$, there exists $\mathcal{S}_{L,\delta}$ and $\mathcal{Q}_{\delta}$ as above such that 
\begin{equation}
\begin{aligned}\label{eq:averageDomain}
\mathcal{E}^{\sigma}(\mathbf{v}, \mathcal{S}_{L,\delta}) \leq c\frac{\delta^2}{H^2} \mathcal{E}^{\sigma}(\mathbf{v}, \mathcal{R}_{L,H}), \quad \mathcal{E}^{\sigma}(\mathbf{v}, \mathcal{Q}_{\delta}) \leq c\frac{\delta^3}{H^2 L} \mathcal{E}^{\sigma}(\mathbf{v}, \mathcal{R}_{L,H}).
\end{aligned}
\end{equation}
Here, the constant $c >0$ is universal. 
\end{proposition}
\noindent We now employ the Poincar\'{e} inequality on the $\delta$-cube to relate the regions of average energy of $\mathbf{v}$ to the energy wells given by $K$.  For reference, for any $\mathcal{Q}_{\delta}  =(s, s+ \delta) \times (s', s'+\delta) \times (s'', s'' + \delta)$ and $\mathbf{u} \in W^{1,1}(\mathcal{Q}_{\delta},\mathbb{R}^3)$, the Poincar\'{e} inequality has the form
\begin{equation}
\begin{aligned}\label{eq:Poincare}
\|\mathbf{u} - \bar{\mathbf{u}} \|_{L^1(\mathcal{Q}_{\delta})} \leq  c_1 \delta \| \nabla \mathbf{u}\|_{L^1(\mathcal{Q}_{\delta})} 
\end{aligned}
\end{equation}
for $\bar{\mathbf{u}} = \fint \mathbf{u}(\mathbf{x}) dx$ (the average) and $c_1$ the uniform Poincar\'{e} constant on a unit cube.  The $\delta$-dependence in the inequality here is key.
\begin{proposition}\label{secondProp}
Let $\mathbf{v} \in \mathcal{M}_{\alpha}$, let $\delta \in (0, \min\{H, L\}]$, and fix a $\delta$-cube $\mathcal{Q}_{\delta}$ as in Proposition \ref{firstProp}.  Then, there exists an $\mathbf{F} \in K$ and $\mathbf{d} \in \mathbb{R}^3$ such that 
\begin{equation}
\begin{aligned}\label{eq:keyEst1}
\int_{\mathcal{Q}_{\delta}} |\mathbf{v} - \mathbf{F}\mathbf{F}_{\alpha}^{-1} \mathbf{z} - \mathbf{d}| dz \leq  \frac{c}{\alpha^{4}}\Big( \frac{\delta^5}{\hat{\sigma} H^2L} \mathcal{E}^{\hat{\sigma}}(\mathbf{v}, \mathcal{R}_{L,H})  + \alpha^3 \frac{\delta^4}{HL^{1/2}} \big(\mathcal{E}^{\hat{\sigma}}(\mathbf{v}, \mathcal{R}_{L,H})\big)^{1/2}\Big).
\end{aligned}
\end{equation}
Here, the constant $c$ is universal.
\end{proposition}
\noindent We now relate regions on which $\mathbf{v}$ has average energy to the hard boundary conditions on the left and right of the domain.  Here, we obtain a quantitative estimate on the closeness of $\mathbf{v}$ to a homogenous deformation $(\mathbf{Id} + \mathbf{b} \otimes \mathbf{m}) \mathbf{F}_{\alpha}^{-1} \mathbf{z}$ in these regions.  This makes use of the fact that we have hard boundary conditions on \textit{both} sides.
\begin{proposition}\label{thirdProp}
Let $\mathbf{v} \in \mathcal{M}_{\alpha}$, let $\delta \in (0, \min\{H, L\}]$, and fix a $\delta$-strip $\mathcal{S}_{L,\delta}$ and corresponding $\delta$-cube $\mathcal{Q}_{\delta}$ as in Proposition \ref{firstProp}.  Then, 
\begin{equation}
\begin{aligned}\label{eq:keyEst2}
&\int_{\mathcal{Q}_{\delta}} |\boldsymbol{\nu} \cdot \big(\mathbf{v} -(\mathbf{Id} + \mathbf{b} \otimes \mathbf{m}) \mathbf{F}_{\alpha}^{-1} \mathbf{z}\big)| dz \leq \frac{c}{\alpha} \Big(  \frac{\delta^3 L^{1/2}}{H}\big( \mathcal{E}^{\hat{\sigma}}(\mathbf{v}, \mathcal{R}_{L,H})\big)^{1/2}\Big).
\end{aligned}
\end{equation}
Here, the constant $c$ is universal, and the direction $\boldsymbol{\nu} \in \mathbb{S}^2$ is as defined Section \ref{ssec:LBSec}.
\end{proposition}
\noindent Finally, we can relate the the two quantities being estimated in (\ref{eq:keyEst1}) and (\ref{eq:keyEst2}) via the inequality 
\begin{equation}
\begin{aligned}\label{eq:keyEstPoint}
&|\boldsymbol{\nu} \cdot \big( \mathbf{F}  \mathbf{F}_{\alpha}^{-1} \mathbf{z} + \mathbf{d} - (\mathbf{I} + \mathbf{b} \otimes \mathbf{m})\mathbf{F}_{\alpha}^{-1} \mathbf{z}  \big)|   \leq |\boldsymbol{\nu} \cdot \big(\mathbf{v} -(\mathbf{Id} + \mathbf{b} \otimes \mathbf{m}) \mathbf{F}_{\alpha}^{-1} \mathbf{z}\big)| +  |\mathbf{v} - \mathbf{F} \mathbf{F}_{\alpha}^{-1} \mathbf{z} - \mathbf{d}|,
\end{aligned}
\end{equation}
and the former, integrated on a $\delta$-cube, can be bounded from below.  
\begin{proposition}\label{fourthProp}
Let $\delta \in (0, \min\{H, L\}]$, $\mathcal{Q}_{\delta} =  (s, s+ \delta) \times (s', s'+\delta) \times (s'', s'' + \delta)$ and $\boldsymbol{\nu}$ as above.  Then, 
\begin{equation}
\begin{aligned}\label{eq:lbFinally}
\int_{\mathcal{Q}_{\delta}} |\boldsymbol{\nu} \cdot \big( \mathbf{F}  \mathbf{F}_{\alpha}^{-1} \mathbf{z} + \mathbf{d} - (\mathbf{I} - \mathbf{b} \otimes \mathbf{m}) \mathbf{F}_{\alpha}^{-1} \mathbf{z}  \big)|dz \geq \frac{1}{4} c_{K} \delta^4, \quad \forall \;\; \mathbf{F} \in K, \mathbf{d} \in \mathbb{R}^3,
\end{aligned}
\end{equation}
for $c_{K}> 0$ as defined in the theorem.
\end{proposition}
\noindent The theorem follows from these estimates.
\begin{proof}[Proof of Theorem \ref{lbTheorem}] Let $\mathbf{v} \in \mathcal{M}_{\alpha}$.  For any $\delta \in (0, \min \{H,L\}]$, we obtain a $\delta$-cube $\mathcal{Q}_{\delta}$ on which the energy is no more that average (Proposition \ref{firstProp}).  By the estimates of Proposition \ref{secondProp} and \ref{thirdProp} on this $\delta$-cube, by the inequality in (\ref{eq:keyEstPoint}), and by the estimate in Proposition \ref{fourthProp}, we deduce the inequality: 
\begin{equation}
\begin{aligned}
&4c_{\star \star} c_K  \leq   \frac{\delta}{\hat{\sigma} H^2 L \alpha^{4}} \tilde{\mathcal{E}}^{\hat{\sigma}}(\mathbf{v}, \mathcal{R}_{L,H})  + \frac{1}{H L^{1/2}\alpha} \big(\tilde{\mathcal{E}}^{\hat{\sigma}}(\mathbf{v}, \mathcal{R}_{L,H})\big)^{1/2} +  \frac{L^{1/2}}{\delta  H \alpha}\big( \tilde{\mathcal{E}}^{\hat{\sigma}}(\mathbf{v}, \mathcal{R}_{L,H})\big)^{1/2}
\end{aligned}
\end{equation}
for some universal $c_{\star \star}> 0$ and $c_K >0$ as defined in the theorem.  Since $\delta \in (0, \min \{H, L\}]$, the second term can be bounded from above by the third, yielding the estimate
\begin{equation}
\begin{aligned}\label{eq:lbProof}
&2c_{\star \star} c_K \leq  \frac{\delta}{\hat{\sigma} H^2 L \alpha^{4}} \tilde{\mathcal{E}}^{\hat{\sigma}}(\mathbf{v}, \mathcal{R}_{L,H})   + \frac{L^{1/2}}{\delta H \alpha }\big( \tilde{\mathcal{E}}^{\hat{\sigma}}(\mathbf{v}, \mathcal{R}_{L,H})\big)^{1/2} , \quad \forall \;\; \delta  \in (0, \min \{H, L\}].
\end{aligned}
\end{equation}
 At least one of the two terms in the upper bound above is $\geq c_{\star \star} c_K$.  Consequently, 
\begin{equation}
\begin{aligned}\label{eq:lbProofp}
\tilde{\mathcal{E}}^{\hat{\sigma}}(\mathbf{v}, \mathcal{R}_{L,H}) &\geq c_{\star} H^2 \min \Big\{ c_K \frac{\hat{\sigma} L\alpha^4}{\delta },  c_K^2 \frac{\delta^2 \alpha^2}{L} \Big\}, \quad \forall \;\; \delta  \in (0, \min \{H, L\}]
\end{aligned}
\end{equation}
for $c_{\star} = \min\{ c_{\star \star},  c_{\star \star}^2\}$.  We are free to maximize this lower bound with respect to $\delta \in (0, \min\{H, L\}]$ to make the inequality sharp.  We claim that
\begin{equation}
\begin{aligned}\label{eq:lbProofpp}
\max_{\delta \in (0, \min\{H, L\}]} \min \Big\{ c_{K} \frac{\hat{\sigma} L \alpha^4}{\delta}, c_K^2 \frac{\delta^2  \alpha^2}{L} \Big\} = \min \Big\{c_K^{4/3} \alpha^{10/3} \hat{\sigma}^{2/3}L^{1/3}, c_K^2\alpha^2H \min \{\frac{H}{L},\frac{L}{H} \}   \Big\} .
\end{aligned}
\end{equation}
Indeed, the maximization has two possibilities: either it is obtained by making the two terms in the set in (\ref{eq:lbProofp}) equal for some $\delta \in (0, \min\{H,L\})$; thus, giving the first term in the set on the right in (\ref{eq:lbProofpp}). Or this is impossible for a $\delta \in (0, \min\{H,L\})$ and it is obtained by evaluating the second term in (\ref{eq:lbProofp}) for $\delta = \min \{ H, L\}$; thus, giving the second term in the set on the right in (\ref{eq:lbProofpp}).  In either case, the maximization happens to be a minimize the set in (\ref{eq:lbProofpp}).  It follows that 
\begin{equation}
\begin{aligned}
\tilde{\mathcal{E}}^{\hat{\sigma}}(\mathbf{v}, \mathcal{R}_{L,H}) \geq c_{\star}H^3 \min \Big\{c_K^{4/3} \alpha^{10/3} \frac{\hat{\sigma}^{2/3}L^{1/3}}{H}, c_K^2\alpha^2\frac{H}{L} , c_K^2\alpha^2 \frac{L}{H} \Big\} .
\end{aligned}
\end{equation}
The theorem follows after taking the infimum over all $\mathbf{v} \in \mathcal{M}_{\alpha}$ for the above inequality.   
\end{proof}

\subsubsection{The proofs.} 
We now turn to a proof of the the propositions above.
\begin{proof}[Proof of Proposition \ref{firstProp}-\ref{secondProp}]
The proof of these two propositions is adapted from the localization result in Lemma 3.1 of  \cite{ChanConti2}.  This can be done via a very minor and straightforward modification of the original proof.  As such, we do not reproduce the argument here. 
\end{proof}

\begin{proof}[Proof of Proposition \ref{thirdProp}]
Let $\mathbf{v}, \delta, \mathcal{S}_{\delta,L}, \mathcal{Q}_{\delta}, \boldsymbol{\nu}$ as defined in the proposition and $\beta \stackrel{\rm def.}{=}  |(\mathbf{Id}  + \mathbf{b} \otimes \mathbf{m})\mathbf{n}^{\perp}|$.  Since $\mathbf{F}_{\alpha} \mathbf{n}^{\perp} = \alpha \mathbf{m}$ and $\mathbf{F}_{\alpha}^{-1} \mathbf{m} = \alpha^{-1}\mathbf{n}^{\perp}$, the boundary conditions for $\mathbf{v} \in \mathcal{M}_{\alpha}$ provide that, in the direction $\boldsymbol{\nu}$ (recall the definition in Section \ref{ssec:LBSec}), 
\begin{equation}
\begin{aligned}
\int_{\mathcal{S}_{\delta,L}} \big(\boldsymbol{\nu} \cdot (\nabla \mathbf{v} ) \mathbf{F}_{\alpha} \mathbf{n}^{\perp} - \beta \big) dz &=  \int_{s''}^{s''+\delta} \int_{s'}^{s' + \delta} \int_0^{L}\Big( \partial_1\big(\alpha (\boldsymbol{\nu} \cdot \mathbf{v})\big) - \beta \Big)dz_1 dz_2 dz_3 \\
&= L \int_{s''}^{s''+\delta}  \int_{s'}^{s'+\delta} \Big( \alpha \big( \boldsymbol{\nu} \cdot (\mathbf{Id} + \mathbf{b} \otimes \mathbf{e}_1) \mathbf{F}_{\alpha}^{-1} \mathbf{m} \big)  -\beta \Big) dz_2 dz_3 = 0.
\end{aligned}
\end{equation}
Let $f_{\pm} \stackrel{\rm def.}{=}  \max\{ 0, \pm f\}$.  Since the integration above vanishes, we also have the inequalities
\begin{equation}
\begin{aligned}
\int_{\mathcal{S}_{\delta,L}} \big|\boldsymbol{\nu} \cdot (\nabla \mathbf{v} ) \mathbf{F}_\alpha \mathbf{n}^{\perp} - \beta \big| dz &= 2 \int_{\mathcal{S}_{\delta,L}} \big(\boldsymbol{\nu} \cdot (\nabla \mathbf{v} ) \mathbf{F}_{\alpha} \mathbf{n}^{\perp} - \beta \big)_{+}dz \leq 2 \int_{\mathcal{S}_{\delta,L}} \big(| (\nabla \mathbf{v} ) \mathbf{F}_{\alpha} \mathbf{n}^{\perp}| - \beta \big)_{+}dz  \\
&\leq 2 \int_{\mathcal{S}_{\delta,L}} \big(|\mathbf{G} \mathbf{n}^{\perp}| + |(\nabla \mathbf{v}) \mathbf{F}_{\alpha} - \mathbf{G}| - \beta \big)_{+}dz \quad \text{ for all } \mathbf{G} \in K.
\end{aligned}
\end{equation}
Further, $|\mathbf{G} \mathbf{n}^{\perp}| = \beta$ for all $\mathbf{G} \in K$ (due to the compatibility conditions (\ref{compat1}-\ref{compat2})), and so we conclude that
\begin{equation}
\begin{aligned}\label{eq:firstOb}
\int_{\mathcal{S}_{\delta,L}} \big|\boldsymbol{\nu}  \cdot (\nabla \mathbf{v} )\mathbf{F}_{\alpha} \mathbf{n}^{\perp} - \beta \big| dz &\leq 2 \int_{\mathcal{S}_{\delta,L}} {\rm dist}\big((\nabla \mathbf{v}) \mathbf{F}_{\alpha} , K\big) dz 
\leq c \frac{L^{1/2} \delta^2}{H} \big(\tilde{\mathcal{E}}^{\hat{\sigma}}(\mathbf{v}, \mathcal{R}_{L,H})\big)^{1/2}.
\end{aligned}
\end{equation}
The latter is due to H\"{o}lder's inequality, the lower bound on the energy density $\hat{\varphi}$, and the first estimate of Proposition \ref{firstProp}.  Now, observe that $\boldsymbol{\nu} \cdot \mathbf{v}(0,z_2,z_3) + \alpha^{-1} \beta z_1 = \boldsymbol{\nu} \cdot (\mathbf{Id} + \mathbf{b} \otimes \mathbf{m}) \mathbf{F}_{\alpha}^{-1} \mathbf{z}$ since $\mathbf{v} \in \mathcal{M}_{\alpha}$.   Consequently, 
\begin{equation}
\begin{aligned}\label{eq:secondOb}
\int_{\mathcal{Q}_{\delta}} | \boldsymbol{\nu} \cdot \big(\mathbf{v}(\mathbf{z}) - (\mathbf{Id} + \mathbf{b} \otimes \mathbf{m})\mathbf{F}_{\alpha}^{-1} \mathbf{z} \big)|dz  &= \int_{\mathcal{Q}_{\delta}} \Big| \alpha^{-1} \int_{0}^{z_1}  (\boldsymbol{\nu} \cdot (\nabla \mathbf{v} (s, z_2,z_3))\mathbf{F}_{\alpha} \mathbf{n}^{\perp} -\beta ) ds \Big| dz \\
&\leq  \frac{\delta}{\alpha}  \int_{\mathcal{S}_{\delta,L}}| \boldsymbol{\nu} \cdot (\nabla \mathbf{v}(\mathbf{z}) ) \mathbf{F}_{\alpha} \mathbf{n}^{\perp} - \beta \big | dz. 
\end{aligned}
\end{equation}
The estimate in the proposition follows by combining (\ref{eq:firstOb}) with (\ref{eq:secondOb}).
\end{proof}
\begin{proof}[Proof of Proposition \ref{fourthProp}]
Let $\delta$, $\mathcal{Q}_{\delta}$ and $\boldsymbol{\nu}$ as in the proposition.  To prove this result, we repeatedly use the fact the $\inf_{\eta \in \mathbb{R}} \int_{s}^{s+\delta} |t \xi + \eta|dt = \frac{\delta^2}{4}|\xi|$.  In particular, 
\begin{equation}
\begin{aligned}\label{eq:lbck}
&\int_{\mathcal{Q}_{\delta}}   |\boldsymbol{\nu} \cdot \big( \mathbf{F} \mathbf{F}_{\alpha}^{-1} \mathbf{z} + \mathbf{d} - (\mathbf{I} + \mathbf{b} \otimes \mathbf{m}) \mathbf{F}_{\alpha}^{-1} \mathbf{z}  \big)|dz\\
&\qquad  \geq \int_{s''}^{s''+\delta} \int_{s'}^{s'+ \delta} \inf_{\eta \in \mathbb{R}}\big | \alpha (\boldsymbol{\nu} \cdot \mathbf{F} \mathbf{n}^{\perp} - |(\mathbf{Id} + \mathbf{b} \otimes \mathbf{m}) \mathbf{n}^{\perp}|)z_1 + \eta \big|dz_1 dz_2 dz_3 \\
&\qquad = \frac{\delta^4}{4}   \alpha \big|\boldsymbol{\nu} \cdot \mathbf{F} \mathbf{n}^{\perp} - |(\mathbf{Id} + \mathbf{b} \otimes \mathbf{m} ) \mathbf{n}^{\perp}|\big|,
\end{aligned}
\end{equation}
and repeating this argument but exchanging the roles of $z_1$ and $z_i$ ($i = 2,3$) leads to the other estimates in the set defining $c_K$ in (\ref{eq:cKDef}). We then trivially conclude the lower bound (\ref{eq:lbFinally}) in the proposition for $c_K$ as defined.  It still remains, however, to prove that $c_K >0$ under for the crystallographic theory and under the additional hypothesis $(\mathbf{I} + \mathbf{b} \otimes \mathbf{m}) \mathbf{n}^{\perp} \cdot \mathbf{a} \neq 0$.  Since the minimization is over a compact set $K$, we need only to show that all three terms in the set cannot be zero simultaneously.  To this end, notice that the first term in the set for $c_K$ (i.e., the one reflected in the lower bound in (\ref{eq:lbck})) is zero if and only if $\mathbf{G} \in K$ satisfies $\boldsymbol{\nu} \cdot \mathbf{G} \in \{ \boldsymbol{\nu} \cdot \mathbf{A} , \boldsymbol{\nu} \cdot \mathbf{B}\}$.  But, in assuming $\boldsymbol{\nu} \cdot \mathbf{G}$ is one of these cases, we find that the second term in the set for $c_K$ is non-vanishing  since 
\begin{equation}
\begin{aligned}
|\boldsymbol{\nu} \cdot ( \{ \mathbf{A}, \mathbf{B}\} - \mathbf{Id}) \mathbf{m}^{\perp}| = |\{ -\lambda, 1-\lambda\} (\mathbf{a} \cdot \boldsymbol{\nu})(  \mathbf{n} \cdot \mathbf{m}^{\perp})| \neq 0.
\end{aligned}
\end{equation}
Note, this quantity is non-zero since the crystallographic theory assumes $\lambda \in (0,1)$, $\alpha \neq 0$, and the added hypothesis above gives $(\mathbf{a} \cdot \boldsymbol{\nu}) \neq 0$.  Therefore, $c_K >0$ as asserted.  
\end{proof}

We finally prove the lower bound in (\ref{eq:EintBound}) with the following observation regarding tensor norms.
\begin{proposition}\label{appendProp}
Let $\mathbb{A} = a_{ijk} \mathbf{e}_i \otimes \mathbf{e}_j \otimes \mathbf{e}_k \in \mathbb{R}^{3\times3\times3}$ represent a third order tensor (with repeated indices summed here and below and for the orthonormal basis $\{\mathbf{e}_1, \mathbf{e}_2, \mathbf{e}_3\} = \{ \mathbf{m}, \mathbf{m}^{\perp}, \mathbf{e}_3\}$).  Then, for $\mathbf{F}_{\alpha}$ as defined in (\ref{eq:tFDef}),
\begin{equation}
\begin{aligned}
\big(a_{ilm}( F_{\alpha})_{lj} (F_{\alpha})_{mk}\big)\big(a_{il'm'} (F_{\alpha})_{l'j} (F_{\alpha})_{m'k}\big) \geq c \alpha^4 |\mathbb{A}|^2,
\end{aligned}
\end{equation}
where $c> 0$ is a universal constant.
\end{proposition}
\begin{proof}
We let $\mathbf{A}_i \stackrel{\rm def.}{=}  a_{ijk} \mathbf{e}_j \otimes \mathbf{e}_k$, and note that it is easy to explicitly verify that 
\begin{equation}
\begin{aligned}
\big(a_{ilm} (F_{\alpha})_{lj} (F_{\alpha})_{mk}\big)(a_{il'm'} (F_{\alpha})_{l'j} (F_{\alpha})_{m'k}\big)  = \sum_{i = 1,2,3} |\mathbf{F}_{\alpha}^T \mathbf{A}_i \mathbf{F}_{\alpha} |^2.
\end{aligned}
\end{equation}
Thus, using a standard estimate for the normed product of matrices $\mathbf{A}, \mathbf{B} \in \mathbb{R}^{3\times3}$, i.e., $|\mathbf{A} \mathbf{B}|^2 \geq \frac{1}{3}\sigma_{min}(\mathbf{B})^2|\mathbf{A}|^2$, we conclude that 
\begin{equation}
\begin{aligned}
\big(a_{ilm} (F_{\alpha})_{lj} (F_{\alpha})_{mk}\big)\big(a_{il'm'} (F_{\alpha})_{l'j} (F_{\alpha})_{m'k}\big) \geq  \sum_{i = 1,2,3} \frac{1}{9} \sigma_{min}^4(\mathbf{F}_{\alpha}) |\mathbf{A}_i|^2 = \frac{1}{9} \sigma_{min}^4(\mathbf{F}_{\alpha})  |\mathbb{A}|^2.
\end{aligned}
\end{equation}
Given the structure of $\mathbf{F}_{\alpha}$, it follows that $\sigma_{min}(\mathbf{F}_{\alpha}) \geq c \alpha $ for some universal $c>0$ as desired.  The proposition follows.
\end{proof}

\end{document}